\documentclass[11pt]{article}
\usepackage[dvipsnames,svgnames,x11names,hyperref]{xcolor}
\usepackage[utf8]{inputenc}
\usepackage{amsthm}
\usepackage{amsmath}
\usepackage{amssymb}
\usepackage[colorlinks]{hyperref}
\usepackage{pgf,tikz}
\usepackage{mathrsfs}
\usetikzlibrary{arrows}
\usepackage{palatino}
\usepackage{mdframed}
\usepackage{xspace}

\usepackage{enumitem}
\usepackage{stfloats}
\usepackage{framed}
\usepackage{color,fancybox,graphicx,url,subfigure}
\usepackage{nicefrac}
\usepackage{tcolorbox}
\usepackage[capitalise]{cleveref}

\def\showauthornotes{1}

\def\showdraftbox{1}

\definecolor{Mygray}{gray}{0.8}

 \ifcsname ifcommentflag\endcsname\else
  \expandafter\let\csname ifcommentflag\expandafter\endcsname
                  \csname iffalse\endcsname
\fi

\ifnum\showauthornotes=1
\else
\fi

\ifnum\showauthornotes=1
\newcommand{\Authornote}[2]{{\sf\small\color{CornflowerBlue}{[#1: #2]}}}
\newcommand{\Authoredit}[2]{{\sf\small\color{red}{[#1]}\color{blue}{#2}}}
\newcommand{\Authorcomment}[2]{{\sf \small\color{Mygray}{[#1: #2]}}}
\newcommand{\Authorfnote}[2]{\footnote{\color{red}{#1: #2}}}
\newcommand{\Authorfixme}[1]{\Authornote{#1}{\textbf{??}}}
\newcommand{\Authormarginmark}[1]{\marginpar{\textcolor{red}{\fbox{%\Large
#1:!}}}}
\else
\newcommand{\Authornote}[2]{}
\newcommand{\Authoredit}[2]{}
\newcommand{\Authorcomment}[2]{}
\newcommand{\Authorfnote}[2]{}
\newcommand{\Authorfixme}[1]{}
\newcommand{\Authormarginmark}[1]{}
\fi

\ifnum\showdraftbox=1

\else

\fi

\newcommand\p{\mbox{\bf P}\xspace}
\newcommand\np{\mbox{\bf NP}\xspace}

\newcommand\rp{\mbox{\bf RP}\xspace}

\newcommand\dtime{\mbox{\bf DTIME}}

\newtheorem{theorem}{Theorem}[section]
\newtheorem*{theorem*}{Theorem}

\newtheorem{definition}[theorem]{Definition}
\newtheorem{lemma}[theorem]{Lemma}

\newtheorem{corollary}[theorem]{Corollary}
\newtheorem{claim}[theorem]{Claim}
\newtheorem{fact}[theorem]{Fact}

\newcommand\Z{\mathbb Z}

\newcommand\R{\mathbb R}
\newcommand\C{\mathbb C}

\newcommand{\V}[1]{\mathbf{#1}}

\newcommand\calS{\mathcal{S}}

\newcommand\calH{\mathcal{H}}

\newcommand\calA{\mathcal{A}}
\newcommand{\calV}{\mathcal{V}}
\newcommand{\calE}{\mathcal{E}}

\renewcommand\Pr{\mathop{\mbox{\bf Pr}}}

 {
	\begin{enumerate}}{\end{enumerate}}

\newcommand{\E}{\mathop{\mbox{\bf E}}}

\newcommand{\suchthat}{\;\ifnum\currentgrouptype=16 \middle\fi|\;}
\newcommand{\defeq}{\mathrel{\mathop:}=}

\hypersetup{linkcolor=blue,citecolor=blue,filecolor=blue,urlcolor=blue}

\newcommand{\KH}[3]{\ensuremath{\texttt{RH}^{#2}_{#1}(#3)}}

\DeclareMathOperator{\Rainbow}{\textsc{Rainbow}}
\DeclareMathOperator{\AlmostRainbow}{\textsc{AlmostRainbow}}

\title{Improved Inapproximability of Rainbow Coloring}

\author{Per Austrin\thanks{\texttt{austrin@kth.se}. Research funded by Swedish Research Council grant 621-2012-4546 and the Approximability and Proof Complexity project funded by the Knut and Alice Wallenberg Foundation.}\\
 KTH Royal Institute of Technology
\and
 Amey Bhangale\thanks{{\tt amey.bhangale@weizmann.ac.il}. Research  supported by Irit Dinur's ERC-CoG grant 772839.}\\
 Weizmann Institute of Science
\and
 Aditya Potukuchi\thanks{{\tt aditya.potukuchi@cs.rutgers.edu}}\\
 Rutgers University
}

\begin{document}

\maketitle
%
%\draftbox

\begin{abstract}

%\Abnote{Right now, after reading the abstract one doesn't get exactly what we are proving without looking at the appropriate setting of parameters. Shall we state in terms of the main corollaries?}
A {\em rainbow} $q$-coloring of a $k$-uniform hypergraph is a $q$-coloring of the vertex set such that every hyperedge contains all $q$ colors.

We prove that given a rainbow $(k - 2\lfloor \sqrt{k}\rfloor)$-colorable $k$-uniform hypergraph, it is $\np$-hard to find a normal $2$-coloring.  Previously, this was only known for rainbow $\lfloor k/2 \rfloor$-colorable hypergraphs (Guruswami and Lee, SODA 2015).

We also study a generalization which we call rainbow $(q,
p)$-coloring, defined as a coloring using $q$ colors such that every
hyperedge contains at least $p$ colors.  We prove that given a rainbow
$(k - \lfloor \sqrt{kc} \rfloor,  k- \lfloor3\sqrt{kc} \rfloor)$-colorable $k$ uniform hypergraph, it is $\np$-hard to find a normal
$c$-coloring for any $c = o(k)$.

The proof of our second result relies on two combinatorial theorems. One of the theorems was proved by Sarkaria (J.~Comb.~Theory.~1990) using topological methods and the other theorem we prove using a generalized Borsuk-Ulam theorem.
\end{abstract}

% !TeX root = rainbow.tex

\section{Introduction}

A $k$-uniform hypergraph $H = (V, E)$ consists of a set of vertices $V$
and a collection $E$ of $k$-element
subsets of $V$, called hyperedges.  A (proper) $c$-\emph{coloring} of $H$ is a coloring of $V$ using $c$ colors such that every hyperedge is non-monochromatic.
The complexity of coloring a hypergraph with few colors has been extensively studied over the years.

For $k=2$ (i.e., graphs), it is $\np$-hard to find a
$3$-coloring whereas finding a $2$-coloring is easy.  For
higher uniformity $k \ge 3$, even finding a $2$-coloring is
$\np$-hard.
From the upper bounds side, given a $3$-colorable graph or
$2$-colorable $3$-uniform hypergraph, the best approximation
algorithms, despite a long line of work~\cite{KNS01, C07, CS08}, only find colorings using
$O(n^\delta)$ colors for some constant $\delta > 0$.

At the same time, strong
inapproximability results for coloring have been elusive.  Given a
$3$-colorable graph, it is $\np$-hard to find a $4$-coloring
\cite{KLS00}, and assuming the {\large $\ltimes$}-Conjecture (a variant of the
Unique Games Conjecture) it is hard to find a coloring using any
constant number of colors \cite{DMR09}.  For large constant $c$, it is known that
it is $\np$-hard to color a $c$-colorable graph using
$2^{\Omega(c^{1/3})}$ colors \cite{Huang13},
and in general it is known that the chromatic number is $\np$-hard to
approximate within $n^{1-\epsilon}$ for every $\epsilon > 0$
\cite{FK98,Zuckerman07}.

In the hypergraph case, stronger hardness results are known: for
instance, given a $4$-colorable $4$-uniform hypergraph or a
$2$-colorable $8$-uniform hypergraph, it is $quasi$-$\np$-hard\footnote{there exists no polynomial time algorithm unless  $\np\subseteq \dtime\left(2^{\log^{O(1)} n}\right)$} to find a
coloring using $2^{(\log n)^{1/20-\epsilon}}$ colors for every
$\epsilon > 0$ \cite{Varma16} following a series of recent
developments \cite{DG13,GHHSV17,Huang15,KhotS17}.  In the $3$-uniform
case, the current best hardness is that given a $3$-colorable
$3$-uniform hypergraph it is $quasi$-$\np$-hard to find a coloring with
$(\log n)^{\gamma / \log \log \log n}$ colors for some $\gamma > 0$
\cite{GHHSV17}. Stronger results are known when the hypergraph is only guaranteed to be \emph{almost} $2$-colorable: given an almost $2$-colorable $4$-uniform hypergraph, it is $quasi$-$\np$-hard to find an independent set of relative size $2^{-\log ^{1-o(1)} n}$ \cite{KS14}.

Given the strong hardness of hypergraph coloring, it is natural to
consider restricted forms of coloring having some additional structure
that might make them more amenable to algorithms.
One such variant is \emph{rainbow colorability} which is introduced in \cite{AGH17}. A $q$ coloring of the hypergraph is called a rainbow $q$-coloring if there exists a coloring of the vertices with $q$ colors such that every hyperedge contains all $q$ colors. More formally,
\begin{definition}[Rainbow Coloring]
  A $q$-coloring $\chi:V \rightarrow [q]$ of a hypergraph
  $H = (V,E)$ is a \emph{rainbow $q$-coloring} if for every
  hyperedge $e \in E$, $\chi^{-1}(e) = [q]$.
\end{definition}

A hypergraph is called rainbow $q$-colorable if there exists a rainbow $q$-coloring. If we restrict the uniformity of the hypergraph to $k$ then the definition of $q$-rainbow coloring is meaningful only when $2\leq q\leq k$. It is easy to observe that the property of $H$ being rainbow $q$-colorable is stronger
the larger $q$ is, and that it is always stronger than
$2$-colorability. We have the following implications on the structure of hypergraphs:
\[
\text{$k$-RC} \Rightarrow \text{$(k-1)$-RC} \Rightarrow \ldots \Rightarrow \text{$2$-RC} \Leftrightarrow \text{$2$-C} \Rightarrow \text{$3$-C} \Rightarrow \ldots \Rightarrow \text{$n$-C},
\]
where $i$-RC stands for ``$H$ is rainbow $i$-colorable'' and $i$-C stands for ``$H$ is $i$-colorable''.

Since rainbow $q$-colorable hypergraphs have more structure than $2$-colorable hypergraphs for $q>2$, one can hope to improve on the known upper bounds on the hypergraph coloring results in \cite{KNS01} when the given hypergraph is rainbow $q$-colorable. In this work, we study the inapproximability of coloring such hypergraphs. More concretely, we study the following problem: what guarantee (in terms of rainbow
$q$-colorability) on $H$ is necessary in order for us to be able (in
polynomial time) to certify that it is $c$-colorable?  Conversely, for what rainbow colorability guarantees is it still $\np$-hard to find a normal $c$-coloring?  More formally,
we define the following decision problem:
\begin{definition} [$\Rainbow(k, q, c)$, $q \le k$]
Given a $k$-uniform hypergraph $H$, distinguish between the following two cases:
\begin{description}
\item[Yes:] $H$ is rainbow $q$-colorable.
\item[No:] $H$ is not $c$-colorable.
\end{description}
\end{definition}
Note that this problem gets \emph{easier} when $q$ increases for a fixed $c$ as well as when $c$ increases for a fixed $q$.

\subsection{Related work}
From the upper bounds side, $\Rainbow(k, k, 2)$ is known to be in $\p$ -- a simple randomized algorithms shows that it is in $\rp$ \cite{McD93} and the problem can be solved without randomness using an SDP \cite{GL15}. In fact, a stronger result is possible: If a given hypergraph is $c$ colorable with the property that there exists two colors, say {\em red}, {\em blue}, such that all the hyperedges contains equal number of red and blue vertices, then the $2$-coloring of such hypergraph can be found in polynomial time.

On the inapproximability side, Guruswami and Lee \cite{GL15} showed that, for all constants
$k, c \ge 2$, $\Rainbow(k, \lfloor k/2 \rfloor, c)$ is $\np$-hard.  Even in the case of $c=2$, this remains the current best $\np$-hardness result in terms of rainbow coloring guarantee for any fixed $k>3$ i.e their result does not rule out $\Rainbow(k, \lfloor k/2 \rfloor+1, 2) \in \p$. \cite{AGH17} asked the question whether it is $\np$-hard to find a $2$-coloring of rainbow $(k-1)$-colorable $k$-uniform hypergraph. Brakensiek and Guruswami \cite{BG16} conjectured
that $\Rainbow(k, k-1, 2)$ is $\np$-hard.  Later they showed
\cite{BG17} that a strong form of this conjecture would follow
assuming a ``V label cover'' conjecture.  Assuming that conjecture, for any $\epsilon > 0 $ it
is $\np$-hard to even find an independent set of an $\epsilon$
fraction of vertices (and in particular it is hard to find a
$1/\epsilon$-coloring) in a rainbow $(k-1)$-colorable $k$-uniform
hypergraph.  However, the V label cover conjecture (which is
essentially a variant of the Unique Games Conjecture with perfect
completeness) is very strong and it is not clear yet whether it should
be believed.

Recently Guruswami and Saket \cite{GS17}, further restrict the guarantee on the rainbow coloring to {\em balanced} rainbow coloring. More specifically, for $Q, k\ge 2$, suppose we are given a $Qk$ uniform hypergraph with the guarantee that it is rainbow $k$-colorable such that every hyperedge $\ell$ colors occur exactly $Q-1$ times, $\ell$ colors occur exactly $Q+1$ and the rest $k-2\ell$ occur exactly $Q$ times for some parameter $1\leq \ell\leq k/2$. In this case, they show that it is $\np$-hard to find an independent set of size roughly $(1-\frac{\ell+1}{k})$. Note that in their result, the hypergraph might not satisfy rainbow $\lfloor k/2 \rfloor+1$-coloring guarantee and therefore the result in~\cite{GS17} does not even rule our efficiently finding $2$-coloring when the $k$-uniform hypergraph is rainbow $\lfloor k/2 \rfloor+1$-colorble.

A dual notion to rainbow colorability is that of \emph{strong
  coloring}. A $k$-uniform hypergraph $H$ is strongly $q$-colorable
for $q \ge k$ if there is a $q$-coloring of $H$ such that every
hyperedge contains $k$ different colors. Note that the two notions coincide when $q=k$.   \cite{BG16} studied the problem of finding a $c$-coloring of a strongly $q$-colorable hypergraph. On the hardness side, they showed that it is $\np$-hard to find a $2$-coloring of a strongly $\lceil 3k/2 \rceil$-colorable $k$-uniform hypergraph. Since the focus of this paper is on rainbow coloring, we refer interested readers to \cite{BG16} for more details about strong rainbow coloring.

\subsection{Our Results}

We show the following hardness results.
First, we give a relatively simple proof that it is $\np$-hard to find
a $2$-coloring even when the graph is guaranteed to be roughly
rainbow $(k-2\sqrt{k})$-colorable.  This improves on the hardness bounds of
\cite{GL15} and settles the smallest previous unknown case which was
$\Rainbow(4,3,2)$.  Concretely, we show the following.

\begin{theorem}
\label{thm:simple}
  For every $t \ge 1, d\geq 2$, $\Rainbow(td+\lfloor \nicefrac{d}{2}\rfloor, t(d-1)+1, 2)$ is $\np$-hard.
\end{theorem}

We have the following corollary (formally proved in \cref{sec:simplegen}):

\begin{corollary}
\label{cor: 2sqrtk}
For all $k\geq 6$, $\Rainbow(k, k- 2\lfloor\sqrt{k}\rfloor, 2)$ is $\np$-hard.
\end{corollary}

The techniques used in the proof the above theorem can only show $2$-coloring in the soundness case.  Towards obtaining similar results for $c > 2$, we introduce a generalization of rainbow coloring in which we only require that each hyperedge contains at least $p$ different colors for some $p \le q$.

\begin{definition}[$(q,p)$-Rainbow Coloring]
  A $q$-coloring $\chi:V \rightarrow [q]$ of a hypergraph
  $H = (V,E)$ is a \emph{rainbow $(q, p)$-coloring} if for every
  hyperedge $e \in E$, $|\chi^{-1}(e)| \geq p$.
\end{definition}

A hypergraph is called rainbow $(q,p)$-colorable if there exists a rainbow $(q, p)$-coloring. Note that rainbow $(q, q)$-coloring is same as rainbow $q$-coloring, and that as long as $p > \lceil q/2 \rceil$ then a $(q,p)$-colorable graph is still always $2$-colorable.
We define the following decision problem analogously to $\Rainbow(k, q, c)$.
\begin{definition} [$\AlmostRainbow(k, q, p, c)$, $p\leq q \le k$, $p > \lceil q/2 \rceil$]
Given a $k$-uniform hypergraph $H$, distinguish between the following two cases:
\begin{description}
\item[Yes:] $H$ is rainbow $(q, p)$-colorable.
\item[No:] $H$ is not $c$-colorable.
\end{description}
\end{definition}

We prove the following hardness result for $\AlmostRainbow(k, q, p, c)$.

\begin{theorem}
  \label{thm:large q}
  For every $d \ge c \ge 2$ and $t \ge 2$ such that $d$ and $t$ are primes and $d$ is odd, let
  $q = t(d - c + 1) + c - 1$ and $k = t d$.
  Then $\AlmostRainbow(k, q, q-d, c)$ is $\np$-hard (provided $d < \lfloor q/2 \rfloor$ so that the $\AlmostRainbow$ problem is well-defined).
\end{theorem}

For $q \ge 4c$, setting $d$ to be a prime between $\sqrt{qc}$ and $2\sqrt{qc}$ we have the following more concrete corollary.
\begin{corollary}
  For infinitely many $q \ge 4c$, $\AlmostRainbow(q + \lfloor \sqrt{qc} \rfloor, q, q - \lfloor 2\sqrt{qc} \rfloor, c)$ is $\np$-hard.
\end{corollary}

In particular this means that $\AlmostRainbow(q + o(q), q, q-o(q),  c)$ is $\np$-hard
for infinitely many $q$ and $c = o(q)$.

A key difference between our results and previous hardness results is
that we only show hardness of finding a $c$-coloring, not hardness of
finding a large independent set (which is an easier task than finding
a $c$-coloring).  In fact, the graphs constructed in our reduction
always have independent sets consisting of almost $1/2$ the vertices.

\subsection{Overview of Proof Ideas}

Like so many other strong hardness of approximation results, our proof
follows the general framework of \emph{long code}-based gadget
reductions from the \emph{label cover} problem.  However, we depart
from the predominant approach of analyzing such reductions using tools
from discrete Fourier analysis such as (reverse) hypercontractivity or
invariance principles.  Indeed, such methods appear inherently
ill-suited to analyze our gadgets -- as alluded to earlier, our
gadgets have very large independent sets, and Fourier-analytic methods
usually can not say anything about the chromatic number of such graphs.

Instead we use methods from topological combinatorics to analyze our
gadgets.  Since its introduction with Lovász' resolution of Kneser's
conjecture in 1978 \cite{LOV78}, topological combinatorics has been
used to resolve a number of combinatorial problems, many of them
regarding the chromatic number of various families of graphs and
hypergraphs.

The lower bound on the chromatic number of Kneser graphs (or more
accurately, the lower bound on the chromatic number of the Schrijver
graphs, which are vertex-critical subgraphs of the Kneser graphs) was
used by Dinur et al.~\cite{DRS02} and recently by Bhangale~\cite{B18} to analyze a long code gadget giving
$\np$-hardness of coloring $3$-uniform hypergraphs with any constant
number of colors and of coloring $4$-uniform hypergraphs with $\mathtt{poly}(\log n)$ number of colors respectively. Apart from these we are not aware of any instance
of results from topological combinatorics being used in hardness of
approximation.

For our results, we construct a new family of hypergraphs that we call
rainbow hypergraphs.  These are $k$-uniform hypergraphs
over the $n$-dimensional $k$-ary cube $[k]^n$, and
$k$ strings $\V x^1, \ldots, \V x^k$ form a hyperedge if, in all but a
constant number $t$ of coordinates $i \in [n]$, it holds that
$\V x^1_i, \ldots, \V x^k_i$ are all different.  Our
hardness results rely on lower bounds on the chromatic number of these
hypergraphs.  For \cref{thm:simple}, a simple direct proof yields
non-$2$-colorability of the corresponding rainbow hypergraph, whereas
for \cref{thm:large q}, we give a proof that the chromatic number of
the corresponding rainbow hypergraph grows with $t$, based on a
generalization of the Borsuk-Ulam theorem (see \cref{thm:gadget bound}).

We now give a brief informal overview of how these rainbow hypergraphs
can be used as gadgets in a label cover reduction.  At their core, these
reductions boil down to a type of \emph{dictatorship testing}, in the
following sense.  We have a large set of functions $f_1, \ldots, f_u: [q]^n
\rightarrow [q]$, and our task is to define a hypergraph with vertex
set $[u] \times [q]^n$ such that
\begin{description}
\item[Completeness] If the functions are all the same dictator function (depending only on one coordinate in their input), then using the function values as colors (i.e., the vertex $(i, \V x)$ gets color $f_i(\V x)$) results in a rainbow $q$-coloring.
\item[Soundness] Each function $f_a$ can be decoded to a small set of coordinates $S_a \subseteq [n]$ (depending only on $f_a$ and not the other functions) such that if the function induces a proper $c$-coloring then many pairs of functions $f_a, f_b$ have overlapping decoded coordinates (i.e., $S_a \cap S_b \ne \emptyset$).
\end{description}

One simple way of constructing such a dictatorship test would be as
follows: let $H$ be a $3$-uniform rainbow hypergraph (over $[3]^n$)
which is not $2$-colorable.  For an edge $\{\V x^1, \V x^2, \V x^3\}$
of $H$, we refer to the set of $\le t$ coordinates where $\{\V x_i^1,
\V x_i^2, \V x_i^3\} \ne [3]$ as the \emph{noisy coordinates} of the
edge.  Now create a $6$-uniform hypergraph on $[u] \times [3]^n$ by
for every pair $a, b \in [u]$ adding an edge consisting of $\{(a, \V
x^1), (a, \V x^2), (a, \V x^3), (b, \V y^1), (b, \V y^2), (b, \V
y^3)\}$ whenever (i) $\{\V x^1, \V x^2, \V x^3\}$ and $\{\V y^1, \V
y^2, \V y^3\}$ are edges in $H$, and (ii) for each $i \in [n]$, $\{\V
x_i^1,\V x_i^2,\V x_i^3,\V y_i^1,\V y_i^2,\V y_i^3\} = [3]$.  It should be
clear that this $6$-uniform graph is $3$-rainbow colorable using any
dictatorship coloring.  For the soundness, consider any $2$-coloring
of the vertices.  By the non $2$-colorability of $H$, each $f_i$ has a
$H$-monochromatic edge $\{\V x^1, \V x^2, \V x^3\}$.  For any pair
$(a,b)$ of such $f$'s with an $H$-monochromatic edge of the same
color, it follows that $\{\V x_i^1,\V x_i^2,\V x_i^3,\V y_i^1,\V y_i^2,\V
y_i^3\} \ne [3]$ for some $i \in [n]$, otherwise we would have a
monochromatic hyperedge.  This means that the set of noisy coordinates
for the two $H$-monochromatic edges overlaps, so if we decode each
$f_a$ to the set of $\le t$ noisy coordinates, then at least half the
pairs of functions $f_a, f_b$ have overlapping decoded coordinates.  This essentially proves hardness of $\Rainbow(6,3,2)$.

To get hardness of $\Rainbow(4,3,2)$, we modify the construction
slightly to make it lopsided by only using one vertex $(b, \V y)$ from
the $b$ part, instead of a full hyperedge of $H$.  It turns out that
the soundness property still holds, using an additional property
that every $2$-coloring of $H$ must have a monochromatic hyperedge
from a large color class.

For the general cases \cref{thm:simple,thm:large q}, the construction
is generalized as follows.  We use as gadget a non-$c$-colorable $d$-uniform rainbow
hypergraph $H$ for $c, d < q$, and construct hyperedges as follows: pick any
$r$ functions $f_{a_1}, \ldots, f_{a_r}$, and for each such $f_{a_j}$
pick $d$ strings $\V x^{j,1}, \ldots, \V x^{j,d} \in [q]^n$ such that
in each coordinate $i \in [n]$, the set of values seen in the $r \cdot
d$ strings is all of $[q]$ (this is the analogue of condition (ii)
above).  The soundness analysis of this construction is more involved.
The key idea here is that for any $\sigma \in {[q] \choose d}$, $f_a$
restricted to $\sigma^n$ induces a coloring of $H$ and thus contains a
monochromatic hyperedge.  If $r$ is sufficiently large, there is in
fact a cover $\sigma_1, \sigma_2, \ldots, \sigma_r \in {[q] \choose
  d}$ of $[q]$ such that the copies of $H$ under each of these
$\sigma_j$'s have a monochromatic hyperedge of the same color.  By
a pigeon hole argument, a constant fraction of $f_a$'s must have the
same monochromatic cover and we show that this can be used to decode each $f_a$ to a
small set of candidate coordinates.

The bound on the uniformity we get is $r \cdot d$, where $r$ is lower
bounded by the need to obtain the covering property described above.
Using a theorem of Sarkaria, we show in \cref{sec:covering bound} that
$r$ can be taken as approximately $\frac{q-c+1}{d-c+1}$
(which is tight for the covering property).

\subsection{Organization}

\Cref{sec:prelims} provides some necessary background material
regarding hardness of Label Cover and a combinatorial covering bound.
In \cref{sec:2coloring gadget} we define the rainbow hypergraph gadget
used for \cref{thm:simple} and show that it is not $2$-colorable.  As
a warmup we then provide in \cref{sec:four} a special case of
\cref{thm:simple}, $\np$-hardness of $\Rainbow(4,3,2)$, since this is
much simpler than the general reductions of \cref{thm:simple,thm:large
  q} (experts may want to skip \cref{sec:four}).  In \cref{sec:general gadget} we define the more general rainbow
hypergraph gadget used for \cref{thm:large q} and lower bound its
chromatic number, and then proceed to prove \cref{thm:large q} in
\cref{sec:generalization}.  The full proof of \cref{thm:simple} and \cref{cor: 2sqrtk} is
given in \cref{sec:simplegen}.  In \cref{sec:conclusion} we give some concluding remarks and further research directions.

\section{Preliminaries}
\label{sec:prelims}

We denote the set $\{1, 2, 3, \ldots, n\}$ by $[n]$. Bold face letters $\V x, \V y,\V z \ldots$ are used to denote strings.  When we have a collection of several strings we use superscripts to index which string is referred to, and subscripts to index into locations in the strings, e.g., $\V x^i_j$ denotes the entry in the $j$'th position of the $i$'th string.

\subsection{Label Cover}
\label{sec:lc}

The starting point in our hardness reductions is the \emph{Layered Label Cover} problem, defined next.

\begin{definition}[Layered Label Cover]
An $\ell$-layed label cover instance consists of $\ell$ sets of variables $X = \{X_1,\ldots, X_\ell\}$. The range of variables in layer $i$ is denoted by $[R_i]$. Every pair of layers $1\leq i < j \leq \ell$ has a set of constraints $\Phi_{ij}$ between the variables in $X_i$ and $X_j$. The constraint between $x\in X_i$ and $y\in X_j$ is denoted by $\phi_{x\rightarrow y}$. Moreover, every constraint between a pair of variables is a projection constraint -- for every assignment $k \in [R_i]$ to $x$ there is a unique assignment to $y$ that satisfies the constraint $\phi_{x\rightarrow y}$.
\end{definition}

In a label cover instance as defined above, for any constraint $\phi_{x \rightarrow y} \in \Phi_{i,j}$, we view it as a function $\phi_{x \rightarrow y}: [R_i] \rightarrow [R_j]$ defined such that for any $k \in [R_i]$, $(k, \phi_{x \rightarrow y}(k))$ satisfies the constraint $\phi_{x \rightarrow y}$. Thus, where there is no ambiguity, we will use $\phi_{x \rightarrow y}$ to denote both the constraint, as well as the function. Moreover, for brevity, we say $x \sim y$, or ``$x$ is a neighbour of $y$'' if $\phi_{x \rightarrow y} \in \Phi_{i,j}$.

\begin{definition}[Weakly dense, \cite{DGKR05}]
  An instance of $\ell$-layered Label Cover is \emph{weakly dense} if the following property holds.
  For any $m$ layers $i_1 < \cdots < i_m$, where $1<m<l$, and any sequence of variable sets $S_k \subseteq X_{i_k}$ for $k \in [m]$ such that $|S_k| \geq \frac{2}{m}|X_{i_k}|$, we have that there are two sets $S_k$ and $S_{k'}$ such that the number of constraints between $S_k$ and $S_{k'}$ is at least a $\frac{1}{m^2}$ fraction of the total number of constraints between layers $X_{i_k}$ and $X_{i_{k'}}$.
\end{definition}

We have the following $\np$-hardness result from \cite{DGKR05}, \cite{DRS02}, which we use as a starting point in proving  Theorem~\ref{thm:simple}.
\begin{theorem}[\cite{DGKR05}, \cite{DRS02}]
\label{theorem:multiLC}
For any constant parameters $\ell \geq 2, r \in \Z$ the following problem is $\np$-hard.
Given a weakly dense $\ell$-layered label cover instance where all variable ranges $[R_i]$ are of size $2^{O(\ell r)}$, distinguish between the following two cases:
\begin{description}
\item[Completeness] There is an assignment satisfying all the constraints of the label cover instance.
\item[Soundness] For every $1 \le i < j \le \ell$, no assignment satisfies more than a $2^{-\Omega(r)}$ fraction of the set of constraints $\Phi_{i,j}$ between layers $i$ and $j$.
\end{description}
\end{theorem}

\subsection{A Covering Bound}
\label{sec:covering bound}

We say a function $f: {[q] \choose d} \rightarrow [c]$ has a
$t$-cover if there is a family $\mathcal{S} \subseteq {[q] \choose d}$
of size $|\mathcal{S}| = t$ such that $\cup_{S \in \mathcal{S}} = [q]$ and $f$ is
constant on $\mathcal{S}$.
Let $B(q,d,c)$ be the minimum $t$ such that every
$f: {[q] \choose d} \rightarrow [c]$ has a $t$-cover.

\begin{claim}
\label{claim:covering lower bound}
For all $1 \le c \le d$, $B(q,d,c) \ge \left \lceil \frac{q-c+1}{d-c+1} \right\rceil$.  For $c \ge d+1$ and $q \ge d+1$ a cover may fail to exist.
\end{claim}

\begin{proof}
  For $S \in {[q] \choose d}$, set $f(S)$ to be the smallest
  $i \in [c-1]$ such that $i \not \in S$, or $f(S) = c$ if
  $[c-1] \subseteq S$.

  By definition, $f^{-1}(i)$ does not cover $[n]$ for $i \in [c-1]$,
  so any cover must use sets from $f^{-1}(c)$.  However all such sets
  contain $[c-1]$, so the total number of elements covered by $k$ sets from $f^{-1}(c)$
  is at most $d + (k-1)(d-c+1)$ thus in order to obtain a cover of all $q$ elements we
  need $d + (k-1)(d-c+1) \ge q$ or equivalently
  $k \ge \frac{q-c+1}{d-c+1}$.
\end{proof}

In the case $c = 2$, there is a simple inductive proof (see \cref{lemma:two}) that the lower
bound of \cref{claim:covering lower bound} is tight.  By a simple
reduction to the \emph{Generalized Kneser Hypergraph}, we get nearly
matching upper bounds for all values of $c$.  The Generalized Kneser
Hypergraph has vertex set $\binom{[n]}{k}$, and a collection of (not necessarily distinct) sets
$\mathcal{S} = \{S_1,\ldots, S_t\}$ forms a hyperedge if each element
in $[n]$ is present in at most $s$ sets in $\mathcal{S}$.  For our
bound, we only need the special case where $s = t-1$, where a
hyperedge just translates to a collection of sets with empty
intersection.

Sarkaria \cite{Sar90} lower bounded the chromatic number of the
Generalized Kneser Hypergraph for many cases, and in particular for
the $s = t-1$ case we have the following.

\begin{theorem}
  \label{thm:sarkaria}
  For any choice of integer parameters $n, k, c, t$ with $n \ge k$ and $t$ prime,
  satisfying $n(t-1) - 1 \ge c(t-1) + t(k-1)$, and any
  $c$-coloring of ${[n] \choose k}$ there exist $t$ sets
  $S_1, \ldots, S_t \in {[n] \choose k}$ of the same color such that
  their intersection is empty.
\end{theorem}

Sarkaria's Theorem as originally stated \cite{Sar90} did not require
$t$ to be prime, but the proof does not work in general for the non-prime case
\cite{LZ07}, and the result is in general currently only known to hold
for $t$ prime or a power of $2$ (see also~\cite{ACCFS18}).
Interestingly enough, all the proofs of the
aforementioned results heavily use topology and we are not aware of
any {\em non-topological} proof of this covering theorem.

Using \cref{thm:sarkaria}, we get a nearly sharp lower bound on
$B(q,d,c)$.  If the requirement that $t$ is prime in
\cref{thm:sarkaria} could be dropped, we would get the exact values of
$B(q,d,c)$.

\begin{theorem}
\label{thm:covgen}
  For all $1 \le c \le d$, $B(q,d,c)  \leq p(q,d,c)$, where $p(q,d,c)$ is the smallest prime that is at least  $\left \lceil \frac{q-c+1}{d-c+1} \right\rceil$.
\end{theorem}

\begin{proof}
  Let $f: {[q] \choose d} \rightarrow [c]$ be arbitrary.  Let $n = q$,
  $k = q-d$, and define
  $\tilde{f}: {[n] \choose k} \rightarrow [c]$ by
  $\tilde{f}(S) = f(\overline{S})$.  By \cref{thm:sarkaria}, for any prime $t$ that satisfies
  $q(t-1) - 1 \ge c(t-1) + t(q-d-1)$, or equivalently
  $t \ge \frac{q-c+1}{d-c+1}$, there exist $t$ sets
  $T_1, \ldots, T_t \in {[n] \choose k}$ such that
  $\cap_{i=1}^{t} T_i = \emptyset$ and
  $\tilde{f}(T_1) = \ldots = \tilde{f}(T_t)$.  Letting
  $S_i = \overline{T_i}$ we have $\cup_{i=1}^{t} S_i = [n]$, so $f$
  indeed has a monochromatic cover of size $t$ provided $t \ge \frac{q-c+1}{d-c+1}$
\end{proof}

% !TeX root = rainbow.tex

\section{Rainbow Hypergraph Gadget for \texorpdfstring{$2$}{2}-coloring}
\label{sec:2coloring gadget}

\newcommand{\KHn}[3]{\mathtt{H}^{{#2}}_{{#1}}({#3})}

\begin{definition} (The hypergraph $\KHn{r}{n}{[d]}$)
\label{def:2col gadget}
Let $\KHn{r}{n}{[d]}$ be the $d$-uniform hypergraph with vertex set $[d]^n$ where $d$ vertices $\V x^1, \ldots, \V x^d \in [d]^n$ form a hyperedge iff
$$\sum_{i=1}^{n} |[d]\setminus \{x^{j}_{i} \mid j\in [d]\}| \leq r $$
The up to $r$ coordinates $i\in [n]$ where $\{x^{j}_{i} \mid j\in [d]\}| \neq d $ are called noisy coordinates.
\end{definition}
In other words, if we write down $\V x^1, \ldots, \V x^d $ in a $d\times n$ matrix form, and it is possible to change up $r$ entries so that all the columns become permutations of $[d]$, then these vertices form a hyperedge.

The following claim shows that the hypergraph $\KHn{r}{n}{[d]}$ is not 2-colorable for $r = \lfloor d/2 \rfloor$.

\begin{lemma}
\label{lemma: ch 3t2t}
For all $d \ge 2$, $\KHn{{\lfloor d/2 \rfloor}}{n}{[d]}$ is not $2$-colorable.
\end{lemma}
\begin{proof}
We prove the claim by induction on $d$. We take the natural convention that the $0$-uniform hypergraph, and a $1$-uniform hypergraph, are not $2$-colorable. Therefore the base cases $d=0$ or $d=1$ are trivial.

Suppose the claim is true for $d-2$. For contradiction assume $\KHn{{\lfloor d/2 \rfloor}}{n}{[d]}$ is $2$-colorable and that $f :  [d]^n \rightarrow \{0,1\}$ is some two coloring of $\KHn{{\lfloor d/2 \rfloor}}{n}{[d]}$. Since $f$ is not a constant function, there exists $\V x^1$ and a coordinate $i$ such that changing $i$'th coordinate of $\V x$ changes the value of $f$. Without loss of generality, $\V x^1 = \V d$, $f(\V x^1) = 1$, and $f(\tilde{\V x}^1) = 0$, where $\tilde{\V x}^1$ is a string which differs from $\V x^1$ only in the $i$'th coordinate.

Now, the restricted function on $[d-1]^n$ cannot be a constant function; since otherwise $\{\V 1, \V 2, \ldots, \V {d-1}\}$ along with either $\V x^1$ or $(\V x^1 +  \delta_{i})$ form a monochromatic hyperedge, contradicting the assumption the $f$ is a proper $2$-coloring of $\KHn{{\lfloor d/2 \rfloor}}{n}{[d]}$. Since, $f$ on $[d-1]^n$ is not a constant function, we can find $\V x^2$ and a coordinate $j$ such that $f(\V x^2)\neq f(\tilde{\V x}^2)$, where again $\tilde{\V x}^2$ differs from $\V x^2$ only at coordinate $j$. Without loss of generality, we can assume $\V x^2 = \V {d-1}$ and $f(\V x^2) = 0$ (and hence  $f(\tilde{\V x}^2) = 1$).

By the induction hypothesis, $\KHn{{\lfloor d/2 \rfloor - 1}}{n}{[d-2]}$ is not $2$-colorable and thus there exists a monochromatic hyperedge if we color the vertices $[d-2]^n$ according to $f$. Let the hyperedge be $\{ \V x^3, \V x^4, \ldots, \V x^d\}$ and  $f(\V x^3) = f(\V x^4) =  \ldots =  f(\V x^d)$. If $f(\V x^3) = 0$, then  $\{ \V x^3, \V x^4, \ldots, \V x^d\} \cup \{ \V x^2, \tilde{\V x}^1\}$ is a $0$-monochromatic hyperedge. Otherwise,  $\{ \V x^3, \V x^4, \ldots, \V x^d\} \cup \{ \V x^1, \tilde{\V x}^2\}$ is a $1$-monochromatic hyperedge.  Thus, $f$ is not a $2$-coloring of $\KHn{{\lfloor d/2 \rfloor}}{n}{[d]}$.

\end{proof}

Let $\alpha(H)$ denote the relative size of a maximum independent set of a hypergraph $H$. We have the following simple fact:

\begin{fact}
\label{fact:khn_indset}
For all $n\geq 2$, $\alpha(\KHn{1}{n}{[3]}) \leq \frac{2}{3}$.
\end{fact}

\section{Warmup: Hardness of \textsc{Rainbow}\texorpdfstring{$(4,3,2)$}{(4,3,2)}}
\label{sec:four}

In this section, we prove the special case of \cref{thm:simple} that $\Rainbow(4,3,2)$ is $\np$-hard. This illustrates many of the ideas of the reductions for the general results in a simpler context, but an expert reader may want to skip this section and instead go directly to the full proof \cref{thm:simple}, in \cref{sec:simplegen}.

\subsection{Reduction}
We give a reduction from the $\ell$-layered label cover instance with parameters $\ell=8$ and $r$ a sufficiently large constant from Theorem~\ref{theorem:multiLC} to a $4$-uniform hypergraph $\calH(\calV, \calE)$. We will select $r$ such that the label cover soundness is smaller than $1/48$. The reduction is given in \cref{fig:red3}.

\begin{figure}[h!]
\begin{tcolorbox}[
    standard jigsaw,
    opacityback=0,
    ]
\paragraph{Vertices $\calV$.} Each vertex $v$ from layer $i$ in the layered label cover instance $\mathcal{L}$ is replaced by a cloud of size $3^{R_i}$ denoted by $C[v] := v\times \{0,1,2\}^{R_i}$. We refer to a vertex from cloud $C[v]$ by a pair $(v,\V x)$ where $\V x\in \{0,1,2\}^{R_i}$. The vertex set of the hypergraph is given by
$$ \calV = \cup_{v\in \cup_i X_i} C[v].$$

\paragraph{Hyperedges $\calE$.}
%Projections are forward
Hyperedges are given by sets $\{(u,\V x),(u,\V y),(u,\V z), (v,\V w)\}$ such that:

\begin{enumerate}
\item There are $i, j$ such that $u \in X_i$, $v \in X_j$, and $u \sim v$.
\item $(\V x, \V y, \V z)$ form an edge in $\KHn{1}{R_i}{\{0,1,2\}}$.
\item \label{item:all colors} $\{\V x_k, \V y_k, \V z_k, \V w_{\phi_{u\rightarrow v}(k)}\} = \{0,1,2\}$ for all $k \in [R_i]$
\end{enumerate}
\end{tcolorbox}
\caption{Reduction to $\Rainbow(4,3,2)$}
\label{fig:red3}
\end{figure}

For a hyperedge $\{(u,\V x), (u,\V y), (u,\V z), (v,\V w)\} \in \calE$, we say that a coordinate $k \in [R_i]$ is \emph{noisy} if $|\{\V x_k, \V y_k, \V z_k \}| = 2$

\begin{lemma}[Completeness]
If the label cover instance is satisfiable then the hypergraph $\calH$ is rainbow $3$-colorable.
\end{lemma}
\begin{proof}
Let $A: \bigcup_i X_i \rightarrow \bigcup_i [R_i]$ define the assignment satisfying all constraints of the layered label cover instance. The rainbow $3$-coloring of the hypergraph is given by assigning a vertex $(v,\V x)$ the color $\V x_{A(v)}$.

A hyperedge $\{(u,\V x),(u,\V y),(u,\V z),(v,\V w)\}$ is thus given the set of colors
\[
\{ \V x_{A(u)}, \V y_{A(u)}, \V z_{A(u)}, \V w_{A(v)} \}.
\]
Since $A$ satisfies all constraints, we have that $A(v) = \phi_{u \rightarrow v}(A(u))$ and by \cref{item:all colors} in the definition of $\calE$ it follows that we see all three colors.
\end{proof}

\begin{lemma}[Soundness]
If the hypergraph $\calH$ is $2$-colorable then there exists an assignment $A$ to the label cover instance which satisfies a $1/48$ fraction of all constraints between some pair of layers $X_i$ and $X_j$.
\end{lemma}
\begin{proof}
Fix a $2$-coloring of the hypergraph. Call the colors red and blue. Consider $\KHn{1}{R_i}{[3]}$ defined on the cloud $C[v]$ for $v \in X_i$. By Lemma~\ref{lemma: ch 3t2t}, and Fact~\ref{fact:khn_indset}, there exists a color class so that more than $\frac{1}{3}$ fraction of vertices in $C[v]$ are colored with that color and there exists a monochromatic hyperedge with the same color. Label a vertex $v$ `red' if that hyperedge is colored red  otherwise label it `blue' (breaking ties using `red' by default). Label a layer with a color which we used to label maximal number of clouds in the layer. Out of the $8$ layers there are at least $4$ layers of the same color. Without loss of generality, let the color be red.

By the weak density property of layered label cover instance, out of these $4$ layers there exist two layers $i$ and $j$ ($i<j$) such that the total number of constraints between the red variables in those two layers is at least $\frac{1}{16}$ times the total number of constraints between $X_i$ and $X_j$. We now give a labeling to the red variables in $X_i$ and $X_j$ which satisfies a constant fraction of the induced constraints.

From now on, let $U$ denote the red variables of $X_i$ and $V$ the red variables of $X_j$. We know from above that the total number of constraints between $U$ and $V$ is at least $\frac{1}{16}$ times the total number of constraints between layers $i$ and $j$. Thus, if we show that we can satisfy a constant fraction of constraints between $U$ and $V$ then we are done.

\paragraph{Labeling.} We define the labeling $A$ to vertices $U \cup V$ as follows: for $u \in U$, the copy of $\KHn{1}{R_i}{[3]}$ has a monochromatic red edge.  Let that edge be $\{(u,\V x), (u, \V y), (u, \V z)\}$.  If the edge has a noisy coordinate $k \in [R_i]$ then set $A(u) = k$, otherwise set $A(u) = 1$.
This defines the labeling of the vertices in $U$.

For $v \in V$, consider the following collection of labels:
$$ S_v = \{ \phi_{u\rightarrow v} (A(u)) \hspace{3pt}|\hspace{3pt} u \in U,~u \sim v \}.$$
Assign a label to $v$ randomly by picking a uniformly random label from $S_v$.

\begin{claim}
\label{claim:setsize}
For every $v \in V$, it holds that $|S_v| \le 3$.
\end{claim}
\begin{proof}
Consider a label $t \in S_v$ where $t\neq 1$. Every such label imposes a restriction on the elements in the cloud $C[v]$ that are colored red.   By definition there is a $u \in U$ such that $A(u) = t \ne 1$, and $\V x,\V y,\V z \in \{0,1,2\}^{[R_i]}$ such that $(u,\V x), (u,\V y), (u,\V z)$ are colored red and $\{ \V x_{A(u)} , \V y_{A(u)}, \V z_{A(u)} \} = 2$. Thus, for every $\V w\in \{0,1,2\}^{R_j}$ such that $(v,\V w)$ is colored red, it must be the case that $\V w_t \in \{ \V x_{A(u)} , \V y_{A(u)}, \V z_{A(u)} \}$ because otherwise $\{(u,\V x), (u,\V y), (u,\V z), (v,\V w)\}$ would form a monochromatic hyperedge of $\calH$.

In other words, for every $t \in S_v \setminus \{1\}$, there is at least one value $z_t \in \{0,1,2\}$ such that all red vertices $(v, \V w)$ of $C[v]$ have $\V w_t \ne z_t$.  This implies that the fraction of red vertices in $C[v]$ is at most $(2/3)^{|S_v|-1}$.  But by construction, at least a $1/3$ fraction of vertices in $C[v]$ are red, and it follows that $|S_v|-1 \le 2$.
\end{proof}

It now follows that the randomized labeling $A$ defined above
satisfies at least a $1/3$ fraction of all constraints between $U$ and
$V$ in expectation, and since the constraints between $U$ and $V$ constitute a $1/16$ fraction of all constraints between layers $i$ and $j$, we are done.
\end{proof}

% !TeX root = rainbow.tex

\section{A Generalized Hypergraph Gadget}
\label{sec:general gadget}

%\subsection{Hypergraph definitions}

In order to prove the hardness of almost rainbow coloring, we will work with the following family of hypergraphs:

\begin{definition}[The hypergraph $\KH{t}{n}{\Sigma}$]
\label{def:khnhyp}
For an alphabet $\Sigma$ of size $p$ and parameters $0 \le t \le n$, let $\KH{t}{n}{\Sigma}$ be the $p$-uniform hypergraph with vertex set $\Sigma^n$ where $p$ vertices $\V x^1, \ldots, \V x^p \in \Sigma^n$ form a hyperedge iff
\begin{equation}
  \label{eqn:KH coord cond}
 |\{ \V x_i^1, \V x_i^2, \ldots, \V x_i^p \} | = p
\end{equation}
for at least $n-t$ different coordinates $i \in [n]$.

The set of \emph{noisy coordinates} for a hyperedge is the set of
$\le t$ values of $i$ where \eqref{eqn:KH coord cond} does not hold.
\end{definition}

The graph $\KH{1}{n}{\{0,1,2\}}$ is very similar to, but not exactly
the same as the hypergraph $\KHn{1}{n}{\{0,1,2\}}$ used in \cref{sec:four}.  The
difference is that in $\KHn{1}{n}{\{0,1,2\}}$, we required the single noisy
coordinate of a hyperedge to have at least $2$ different colors,
whereas in $\KH{1}{n}{\{0,1,2\}}$ the noisy coordinate may have only a
single color.  This difference is mostly superficial, and we could
have defined $\KHn{1}{n}{\{0,1,2\}}$ differently to make it match
$\KH{1}{n}{\{0,1,2\}}$ (but the additional edges contained in
$\KH{1}{n}{\{0,1,2\}}$ would not have been used in the reduction for
$\Rainbow(4,3,2)$).

Note that $\KH{t}{n}{\Z_p}$ has very large ``non-junta-like''
independent sets containing almost half the vertices, e.g.~the set of
all strings containing more than $n/p + t$ zeros is independent and
has size $1/2-o(1)$ for fixed $t$ and $p$ as $n \rightarrow\ \infty$.

Generalizing \cref{lemma: ch 3t2t}, we want to obtain lower
bounds on the chromatic number of $\KH{t}{n}{\Z_p}$ that grow with
$t$.

Our main combinatorial result is the following.
\begin{theorem}
\label{thm:gadget bound}
  For every odd prime $p$ and $c, n \ge 1$, the chromatic number of
  $\KH{p^2c}{n}{\Z_p}$ is at least $c+1$.
\end{theorem}

The proof is given in \cref{sec:schrijver proof}. This bound is likely far from tight (for one thing, note that for
fixed $t$, the value of $c$ even decreases with $p$).

\subsection{Topology Interlude}

In this subsection, we cover some necessary topological notions and theorems that will be used in the proof of \cref{thm:gadget bound}. The curious reader is referred to Matou\v{s}ek's excellent book~\cite{MAT07} for proofs and further details.

We use $S^d = \{\V x \in \mathbb{R}^{d+1} \suchthat \|\V x\| =1\}$ to denote the unit $d$-sphere.

\begin{definition}[Free $\mathbb{Z}_p$-action]
For a topological space $X$, a $\mathbb{Z}_p$-action on $X$ is a collection $\Phi = \{\psi_g\}_{g \in \mathbb{Z}_p}$ of homeomorphisms $X \mapsto X$ such that for every $g \in G$, the map $\psi_g$ is continuous, and for every $g,h \in \mathbb{Z}_p$, we have that $\psi_g \circ \psi_h = \psi_{gh}$. Moreover, the action is \emph{free} is for every nonzero $g \in \mathbb{Z}_p$, and every $\V x \in X$, we have $\psi_g(\V x) \neq \V x$.
\end{definition}

We shall mainly talk about $\mathbb{Z}_p$-actions on a sphere $S^k$, where $p$ is a prime and $k$ is odd. In this case, every nonzero element of $\mathbb{Z}_p$ has essentially the same kind of action, i.e., for every nonzero $g \in \mathbb{Z}_p$, and every $\V x \in S^k$, we have

\begin{enumerate}
\item $\psi_g(\V x) \neq \V x$.
\item $(\psi_g)^p(\V x) = \V x$.
\end{enumerate}

Hence, we shall just pick an arbitrary nonzero element $g$ of $\mathbb{Z}_p$, and define $L \defeq \psi_g$. By slight abuse of notation, we shall call $L$ the free $\mathbb{Z}_p$-action, also since it determines how every other element acts.

Let $\omega_p = \exp(2 \pi i / p)$ be the primitive $p$'th root of
unity in $\C$.  In our uses, $p$ will always be some fixed prime and
we omit the subscript and simply write $\omega$.  Let $\phi: \R^{2n}
\rightarrow \C^n$ be the bijection $\phi(\V x) = (x_{2j-1} + i x_{2j})_{j
  \in [n]}$ (i.e., we clump together pairs of coordinates in
$\R^{2n}$).

\begin{fact}
  \label{fact:Zp-action}
  For every odd prime $p$ and integer $n \ge 1$ the map $L: S^{2n-1} \rightarrow S^{2n-1}$ defined by $L(\V x) = \phi^{-1}(\omega \phi(\V x))$ is a free $\Z_p$-action on $S^{2n-1}$.
\end{fact}

It is important that the sphere in the above fact is an odd sphere as only $\Z_2$ acts freely on even spheres. We use the following generalization of the classic Borsuk-Ulam Theorem.

\begin{theorem}[\cite{WOJ96}, or \cite{MAT07} Theorem 6.3.3]
\label{thm:ZpBU}
Let $p$ be an odd prime, and let $S = S^{(p-1)d + 1}$. Let $f:S \rightarrow \mathbb{R}^d$ be a continuous map, and $L$ be any free $\mathbb{Z}_p$-action on $S$. Then, there is some point $\V x \in S$ such that
\[
f(\V x) = f(L \V x) = f(L^2 \V x) = \cdots = f(L^{p-1}\V x)
\]

\end{theorem}

With the above general theorem at hand, we can draw the same covering conclusion as in the Lusternik-Schnirelmann theorem on covering (see, for example,~\cite{MAT07}, Exercise $6.3.4$).

\begin{corollary}
\label{corr:dcover}
For any covering of $S^{(p-1)(c-1)+1}$ by $c$ closed sets $A_1, \ldots, A_c$, there is an $i \in [c]$ and a point $\V x \in S^{(p-1)(c-1)+1}$ such that $\V x, L\V x, \ldots, L^{p-1}\V x$ are all contained in $A_i$.
\end{corollary}

\subsection{Bound on the chromatic number}
\label{sec:schrijver proof}

In this section we give a lower bound on the chromatic number of $\KH{t}{n}{\Z_p}$.

The proof is basically an adaptation of B\'{a}r\'{a}ny's proof~\cite{BAR78} of Lov\'{a}sz's theorem~\cite{LOV78} on the chromatic number of Kneser graphs. In order to carry this out, one needs to adapt an equivalent formulation of Gales' theorem.

Before proceeding with the proof, we develop some notation that will
be useful.  For an even integer $d$, we have the bijection $\phi: \R^d
\rightarrow \C^{d/2}$ and the free $\Z_p$-action $L$ from
\cref{fact:Zp-action} acting on $S^{d-1}$ by taking $\V z$ to
$\phi^{-1}(\omega \phi(\V z))$.  Define a bilinear function $M: \R^{d}
\times \R^{d} \rightarrow \R^2$ by
\[
M(\V w, \V z) = \phi^{-1}\!\left( \left\langle \phi(\V w), \overline{ \phi(\V z) } \right\rangle \right)
\]
where $\langle \cdot, \cdot \rangle$ is the usual inner product over $\C^{d/2}$ and by slight abuse of notation we view $\phi$ also as a bijection between $\R^2$ and $\C$.
For brevity, we will parameterize this function by the first variable and denote $M_{\V w}(\V z) = M(\V w,\V z)$. The key properties to note are:

\begin{description}
\item[(M1)] $M$ is bilinear and in particular for $L \V z = \phi^{-1}(\omega \phi \V z)$ we have
  \[
  M_{\V w}(L \V z) = \phi^{-1} ( \omega \langle \phi(\V w), \overline{ \phi(\V z) } \rangle )
  \]
  which equals both $M_{L \V w}(\V z)$ and $L M_{\V w}(\V z)$ (where, just like with $\phi$, we view $L$ as also acting on $\R^2$ by rotating every point counter-clockwise by $2 \pi/p$ around the origin).
\item[(M2)]
  For $w \neq \overline {0}$, we have that $M_{\V w}$ is a full rank map, i.e., $\operatorname{image}(M_{\V w}) = \R^2$.
\end{description}

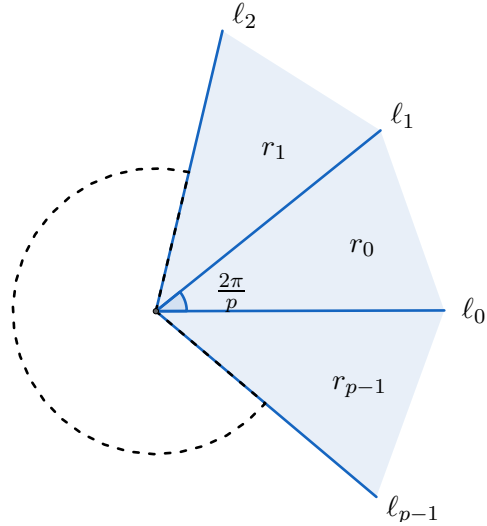
\begin{figure}
  \centering
\definecolor{rvwvcq}{rgb}{0.08235294117647059,0.396078431372549,0.7529411764705882}
\definecolor{wrwrwr}{rgb}{0.3803921568627451,0.3803921568627451,0.3803921568627451}
\begin{tikzpicture}[line cap=round,line join=round,>=triangle 45,x=0.5cm,y=0.5cm]
\clip(-18,-4.6) rectangle (10,10);
\fill[line width=1pt,color=rvwvcq,fill=rvwvcq,fill opacity=0.10000000149011612] (-3,1) -- (4.668011557621793,1.0199687800979735) -- (2.9720581008240856,5.809711222140203) -- cycle;
\fill[line width=1pt,color=rvwvcq,fill=rvwvcq,fill opacity=0.10000000149011612] (-3,1) -- (2.9720581008240856,5.809711222140203) -- (-1.2302696325492632,8.461022344593445) -- cycle;
\fill[line width=1pt,color=rvwvcq,fill=rvwvcq,fill opacity=0.10000000149011612] (-3,1) -- (4.668011557621793,1.0199687800979735) -- (2.8627725672151083,-3.9423372836250223) -- cycle;

\draw [line width=1pt,color=rvwvcq] (-3,1)-- (4.668011557621793,1.0199687800979735);
\draw (4.868011557621793,1.6199687800979735) node[anchor=north west] {$\ell_{0}$};
%\draw [line width=1pt,color=rvwvcq] (4.668011557621793,1.0199687800979735)-- (2.9720581008240856,5.809711222140203);
\draw [line width=1pt,color=rvwvcq] (-3,1)-- (2.9720581008240856,5.809711222140203);
\draw (2.9720581008240856,6.809711222140203) node[anchor=north west] {$\ell_{1}$};
%\draw [line width=1pt,color=rvwvcq] (2.9720581008240856,5.809711222140203)-- (-1.2302696325492632,8.461022344593445);
\draw [line width=1pt,color=rvwvcq] (-1.2302696325492632,8.461022344593445)-- (-3,1);
\draw (-1.2302696325492632,9.461022344593445) node[anchor=north west] {$\ell_{2}$};
%\draw [line width=1pt,color=rvwvcq] (4.668011557621793,1.0199687800979735)-- (2.8627725672151083,-3.9423372836250223);
\draw [line width=1pt,color=rvwvcq] (2.8627725672151083,-3.9423372836250223)-- (-3,1);
\draw (2.8627725672151083,-3.6423372836250223) node[anchor=north west] {$\ell_{p-1}$};
\draw [shift={(-3,1)},line width=1pt,dash pattern=on 2pt off 4pt]  (0,0) --  plot[domain=1.3379035891664772:5.582767172249694,variable=\t]({1*3.8048022545162496*cos(\t r)+0*3.8048022545162496*sin(\t r)},{0*3.8048022545162496*cos(\t r)+1*3.8048022545162496*sin(\t r)}) -- cycle ;
\draw [shift={(-3,1)},line width=1pt,color=rvwvcq,fill=rvwvcq,fill opacity=0.10000000149011612]  (0,0) --  plot[domain=0.002604160779811735:0.6780062367409331,variable=\t]({1*0.8201534689938746*cos(\t r)+0*0.8201534689938746*sin(\t r)},{0*0.8201534689938746*cos(\t r)+1*0.8201534689938746*sin(\t r)}) -- cycle ;
\draw (-1.68,2.28) node[anchor=north west] {$\frac{2\pi}{p}$};
\draw (1.88,3.18) node[anchor=north west] {$r_0$};
\draw (-0.46,5.76) node[anchor=north west] {$r_1$};
\draw (1.42,-0.46) node[anchor=north west] {$r_{p-1}$};
\begin{scriptsize}
\draw [fill=wrwrwr] (-3,1) circle (1pt);
\end{scriptsize}
\end{tikzpicture}
\caption{Definition of the cones $r_0, \ldots, r_{p-1}$}
\label{fig:cones}
\end{figure}

Next, we define a function $T:  \R^{d} \times \R^{d} \rightarrow \{\perp, 0,\ldots, p-1\}$ which is almost like a $p$-way threshold function. Let $r_0, r_1, \ldots, r_{p-1}$ be the open cones as shown in \cref{fig:cones}.
More precisely, denote by $\ell_j\in \mathbb{R}^2$ the ray $\left\{ \left(\alpha \cos(\frac{2\pi j}{p}), \alpha \sin(\frac{2\pi j}{p})\right) \mid \alpha\geq 0\right\}$ for $0\leq j\leq p-1$. With this notation, $r_j$ is an open region between $\ell_j$ and $\ell_{j+1 \bmod p}$ as shown in the figure.
 We define:
\[
T_{\V w}(\V z) = \begin{cases}
j & \text{if  $\quad M_{\V w}(\V z)\in r_j$ for some $j$} \\
\perp & \text{otherwise}
\end{cases}
\]
Note that $T_{\V w}$ almost acts like a threshold function except it does not deal with ``ties'' -- in case of a tie, $T_{\V w}$ is simply defined as $\perp$.  The most important property of $T_{\V w}$ is that it interacts well with $L$:

\begin{claim}
  \label{claim:rotate}
  For all integers $j \ge 0$, and all $\V w, \V z \in \R^{d}$, it holds that
  \[
  T_{L^j \V w}(\V z) = T_{\V w}(L^j \V z) = \begin{cases}
    (T_{\V w}(\V z) + j) \bmod p & \text{if $T_{\V w}(\V z) \ne \perp$} \\
    \perp & \text{otherwise}
  \end{cases}
  \]
\end{claim}
\begin{proof}
  By Property~(M1), $M_{L^j \V w}(\V z) = M_{\V w}(L^j \V z)$ equals $M_{\V w}(\V z)$ rotated $2 \pi j/p$ radians counter-clockwise around the origin.  Thus if $M_{\V w}(\V z) \in r_k$ for some $k$ then $M_{\V w}(L^j \V z) \in r_{k + j \bmod p}$ (and thus $T_{\V w}(L^j \V z) = (k+j) \bmod p$) and similarly if $M_{\V w}(\V z) \in \ell_k$ then $M_{\V w}(L^j \V z) \in \ell_{k + j \bmod p}$ (and thus $T_{\V w}(L^j \V z) = \perp$).
\end{proof}

  Let $\V u: \R_{\ge 0} \rightarrow S^{d-1}$ be the normalized moment curve in $\R^d$, i.e., $\V u(s) = \gamma(s) / \| \gamma(s) \|_2$ where $\gamma(s) = (1,s,s^2 ,\ldots,s^{d-1})$. One important property to note is that for any subset $S \subset \mathbb{R}$ such that $|S| \leq d$, we have that the vectors $\{\V u(s)\}_{s \in S}$ are linearly independent.  We have the following basic fact.

\begin{claim}
  \label{claim:curve}
  For every $\V w \in S^{d-1}$, $T_{\V w}(\V u(s)) = \perp$ for less than $pd$ different values of $s \in \R$.
\end{claim}

\begin{proof}
Suppose for contradiction that at least $pd$ points $M_{\V w}(\V u(s))$ lie on the $p$ rays $\ell_0, \ell_1, \ldots, $ $\ell_{p-1}$. Of these at least $d$ lie on a line. Since any subset of at most $d$ $\V u(s)$'s are in general position, this contradicts Property~(M2) that $\operatorname{image}(M_{\V w}) = \R^2$.
\end{proof}

The choice of $\V u$ is somewhat arbitrary -- any continuous curve whose image under $M_{\V w}$ intersects the $\ell_k$'s in a finite number of points would work.
With these facts in hand, we are ready to prove \cref{thm:gadget bound}.

\begin{theorem*}[\Cref{thm:gadget bound} restated]
  For every odd prime $p$ and $c, n \ge 1$, the chromatic number of
  $\KH{p^2c}{n}{\Z_p}$ is at least $c+1$.
\end{theorem*}

\begin{proof}
  Let $d := (p-1)(c-1) + 2$.
  We construct a set of $n$ points $\calV = \{\V v^1, \V v^2, \ldots, \V v^n\}$ on $S^{d-1}$, one for every index in $[n]$, as follows:
  \[
  \V v^i = L^{i-1} \V u(i)
  \]
  The key property of these points is that they give a correspondence between points in $S^{d-1}$ and the vertices in $\KH{p^2c}{n}{\Z_p}$ (i.e., $\Z_p^n$) in the following sense.
  We say that $\V x \in \Z_p^n$ \emph{matches} $\V w \in S^{d-1}$ if
  \[
 \V x_i = T_{\V w}(\V v^i)
  \]
  for all $i \in [n]$ such that $T_{\V w}(\V v^i) \ne \perp$.
  Now, given a coloring $\chi: \Z_p^n \rightarrow [c]$,
  we define a covering $\{A_1, A_2, \ldots, A_c\}$ of $S^{d-1}$ as
  follows: for every point $\V w \in S^{d-1}$, put $\V w\in A_c$ if
  there is a $\V x \in \Z_p^n$ that matches $\V w$ and
  has $\chi(\V x) = c$.  Observe that it is possible that a point $\V
  a$ belongs to many $A_j$'s and that every point point $\V a \in S^{d-1}$ is matched by at least one $\V x \in \Z_p^n$ (so that this is indeed a cover).

  Next, we observe that the sets $A_1, \ldots, A_c$ are closed.
  \begin{claim}
    Each $A_j$ is closed.
  \end{claim}
  \begin{proof}
    Note that the map $\V w \mapsto M_{\V w}(\V v^i)$ is continuous for each $i \in [n]$.
    Thus for every $\V w \in S^{d-1}$, there is some $\epsilon > 0$ such
    that for every $\V w'$ within distance $\epsilon$ of $w$ it holds that
    \begin{equation}
      \label{eq:A_j limit}
      \text{for every $i \in [n]$, either $T_{\V w'}(\V v^i) = T_{\V w}(\V v^i)$ or $T_{\V w}(\V v^i) = \perp$}
    \end{equation}
    Now let $\V w$ be a point in the closure of $A_j$.  Taking $\epsilon > 0$ as above, there is an $\V w' \in A_j$ within distance $\epsilon$ of $\V w$ satisfying \eqref{eq:A_j limit}.
    But any $\V x$ that matches such an $\V w'$ also matches $\V w$ and in particular it follows that $\V w \in A_j$ and hence $\overline{A_j} = A_j$.
  \end{proof}

  Thus, $\{A_1, \ldots, A_c\}$ is a cover of $S^{d-1} = S^{(p-1)(c-1)+1}$ by $c$ closed sets, so by \cref{corr:dcover} there is a
  point $\V w^\star \in S^{d-1}$ such that $\V w^\star, L\V
  w^\star,\ldots, L^{p-1}\V w^\star$ are all covered by the same
  set. Suppose that this set is $A_1$.  For each $j \in \Z_p$, let $\V
  x^{j}$ be any vertex of $\KH{p^2c}{n}{\Z_p}$ that has $\chi(\V x^{j}) =
  1$ and that matches $L^j \V w^\star$.  By construction these $p$
  vertices have the same color and all that remains to prove is the
  following claim.

  \begin{claim}
    $\V x^{0}, \V x^{1}, \ldots, \V x^{p-1}$ form a hyperedge in $\KH{p^2c}{n}{\Z_p}$
  \end{claim}
  \begin{proof}
    To prove this, it suffices to show that for every $i\in [n]$ such that $T_{\V w^\star}(\V v^i) \neq \perp$, we have $\{\V x^{0}_i, \V x^{1}_i, \ldots, \V x^{p-1}_i\} = \Z_p$, since the number of $i\in [n]$ s.t.~$T_{\V w^\star}(\V v^i) = \perp$ is at most $pd \le p^2c$.
    To prove this, first note that by definition $\V x^j_i = T_{\V L^j w^{\star}}(\V v^i)$ for all $i$ such that $T_{\V w^{\star}}(\V v^i) \ne \perp$.  By \cref{claim:rotate} it thus follows that $\V x^j_i = (\V x^0_i + j) \bmod p$.
  \end{proof}

  Thus any $\chi: V(\KH{p^2c}{n}{\Z_p}) \rightarrow [c]$ must have a monochromatic hyperedge and the proof of \cref{thm:gadget bound} is done.
\end{proof}

% !TeX root = rainbow.tex
\newcommand{\inda}{\zeta}
\newcommand{\indb}{\eta}

\section{Almost Rainbow Hardness}
\label{sec:generalization}

In this section we prove Theorem~\ref{thm:large q}. 
Recall from \cref{sec:covering bound}
that $B(q,d,c)$ is the
worst case covering size $t$ such that every function
$g:{[q]\choose d} \rightarrow [c]$ has a monochromatic cover of size
$t$.

\begin{theorem}[\cref{thm:large q} restated]
For every $d \ge c \ge 2$ and $t \ge 2$ such that $d$ and $t$ are primes and $d$ is odd, let
  $q =t(d - c + 1) + c - 1$ and $k = t d$.
  Then $\AlmostRainbow(k, q, q-d, c)$ is $\np$-hard (provided $d < \lfloor q/2 \rfloor$)
\end{theorem}

In the rest of this section, fix $t:=\frac{q-c+1}{d-c+1}$ which is equal to $B(q,d,c)$ using \cref{thm:covgen}, as $t$ is a prime number for the setting of $q$ in the above theorem.

For this result, we do not need the full power of layered Label Cover, but use \cref{theorem:multiLC} with $\ell = 2$ layers (i.e., normal Label Cover).
To simplify notation in this case, we refer to the two vertex sets as $U = X_1$ and $V = X_2$, and denote the alphabet size of $U$ by $R$ and the alphabet size of $V$ by $L$.
In other words, our starting point is a label cover instance on variables $U \cup V$ with alphabet sizes $R$ and $L$ of size $2^{O(r)}$ and soundness $2^{-\Omega(r)}$ for some parameter $r$ that will be chosen to a large enough constant as a function of $q$, $d$ and $c$ later.

We reduce it to a hypergraph $\calH(\calV, \calE)$ using the reduction given in \cref{fig:red3t}.

\begin{figure}[ht!]
\begin{tcolorbox}[
    standard jigsaw,
    opacityback=0, 
    ]

\paragraph{Vertices $\calV$.} Each vertex $u\in U$ in the Label Cover instance $\mathcal{L}$ is replaced by a cloud of size $q^{R}$ denoted by $C[u] := \{u\}\times [q]^{R}$. We refer to a vertex from the cloud $C[u]$ by a pair $(u, \V x)$ where $\V x\in [q]^{R}$. The vertex set of the hypergraph is given by 
$$ \calV = \cup_{u\in U} C[u].$$

\paragraph{Hyperedges $\calE$.}
%Projections are forward
For every vertex $v\in V$ and every set of $t$ neighbors $u_1, u_2, \ldots, u_t$ of $v$ from $U$, we add the following hyperedges.
Let $\pi_i = \phi_{u_i \rightarrow v}$ be the projection constraint between $u_{i}$ and $v$ for $1\leq i\leq t$.

Let $\V x^{i,j} \in [q]^{R}$ be a set of $t d$ strings indexed by $i \in [d]$ and $j \in [t]$.  If it holds that for every $\beta \in [L]$ and all choices of $\alpha_j \in \pi_j^{-1}(\beta) \subseteq [R]$ for $j \in [t]$ that
\begin{equation}
\label{eq:crosshypedge_cond}
\left|\left\{ \V x^{i,j}_{\alpha_j} \,|\,i \in [d], j \in [t] \right\} \right| \geq q-d
\end{equation}
then we add add the hyperedge
\[
\{(u_j, \V x^{i,j})\}_{i \in [d], j \in [t]} \in {\calV \choose td}
\]
to the hypergraph.
\end{tcolorbox}
\caption{Reduction to $\AlmostRainbow(td, q , q-d, c)$. }
\label{fig:red3t}
\end{figure}

\begin{lemma}[Completeness]
\label{lemma:completeness_gen}
If the Label Cover instance is satisfiable then the hypergraph $\calH$ is $(q, q-d)$-rainbow colorable.
\end{lemma}
\begin{proof}
Let $A: U\cup V \rightarrow [R]\cup [L]$ define the satisfiable labeling to the Label Cover instance. The rainbow $(q, q-d)$-coloring of the hypergraph is given by assigning a vertex $(u,\V x)$ with a color $\V x_{A(u)}$. 

To see that this is a rainbow $(q, q-d)$-coloring, consider any hyperegde in the hypergraph between the clouds $C[u_1], C[u_2], \ldots, C[u_t]$ where $u_1, u_2, \ldots, u_t \in U$ and $v\in V$ be their common neighbor. This hyperegde is of the form 
\[
\{(u_j, \V x^{i,j})\}_{i \in [d], j \in [t]} \in {\calV \choose td}
\]
satisfying the  \eqref{eq:crosshypedge_cond}. By definition, $\chi$ assigns color $\V x^{i, j}_{A(u_j)}$ to vertices $\{(u_j, \V x^{i,j})\}$ for $i\in [d]$ and $j\in [t]$. It is easy to see from \eqref{eq:crosshypedge_cond} that these vertices get $q-d$ distinct colors since $A(u_j) \in \pi_{j}^{-1}(A(v))$ for all $1\leq j\leq t$.

Hence $\chi$ is a valid $(q, q-d)$-rainbow coloring.
\end{proof}

\newcommand{\calI}{\mathcal{I}}

We now prove the main soundness lemma.
\begin{lemma}[Soundness]
\label{lemma:soundness_gen}
If $\calH$ is properly $c$-colorable then there is an assignment $A$ to the Label Cover instance which satisfies an $\frac{1}{d^4c^3 t^4 2^{td\log q}}$ fraction of all constraints between $U$ and $V$.
\end{lemma}
\begin{proof}
Assume for contradiction that the hypergraph $\calH$ is $c$-colorable. Fix a $c$-coloring $\chi : \calV \rightarrow [c]$ of the vertices of $\calH$.

%\item
 Set $h = d^2c$. For every $u\in U$, define functions $f_u : {[q]\choose d} \rightarrow [c]$, and $g_u:{[q] \choose d} \rightarrow 2^{[R]}$ as follows. For a $\sigma \in {[q]\choose d}$, in a cloud $C[u]$, consider the induced $d$-uniform hypergraph $\KH{h}{R}{\sigma}$. Look at the coloring on these vertices induced by $\chi$ i.e.~$\chi_{u, \sigma} : \sigma^{R} \rightarrow [c]$ defined by $\chi_{u, \sigma} (\V x) = \chi((u,\V x))$.  By \cref{thm:gadget bound}, there exists a color class, say $b\in [c]$, such that there exists a monochromatic hyperedge with color $b$ in $\KH{h}{R}{\sigma}$.  Set $f_u(\sigma) = b$, where $b$ is one such color class, breaking ties arbitrarily. Also, set $g_u(\sigma) = J$ if $J\subseteq [R]$ are the set of {\it noisy coordinates} in the $b$-monochromatic hyperedge, again breaking ties arbitrarily. If none of the coordinates are noisy in the hyperedge, then set $g_u(\sigma) = \{1\}$.

%\item 
Recall that $t = B(q,d,c)$, so by definition, for each variable $u$, subsets $\sigma^u_1, \sigma^u_2, \ldots, $ $\sigma^u_t \in {[q] \choose d}$ and a color $b_u \in [c]$ such that $f_u(\sigma^u_j) = b_u$ for all $j \in [t]$ and $\cup_{j=1}^t \sigma^u_j = [q]$.    Write $S_u = (\sigma^u_1,\ldots \sigma^u_t) \in {[q] \choose d}^t$ and label a variable $u$ as $(S_u, b_u)$.
%\item 
Let $T$ be the total number of coverings of ${[q]\choose d}$ of size at most $t$.  A trivial upper bound on $T$ is ${q\choose d}^t \leq 2^{td\log q}$.  By an averaging argument, there is a label $(S, b)$ such that at least a $\nicefrac{1}{cT}$ fraction of all constraints of the Label Cover instance are incident upon vertices $u \in U$ with label $(S, b)$.  Let that subset be $U'$.  Thus, between $U'$ and $V$, we have at least a $\nicefrac{1}{cT}$ fraction of all constraints.

For the rest of the analysis, we focus on satisfying the constraints between $U'$
and $V$.  Let $S = \{\sigma_1, \sigma_2, \ldots, \sigma_t\}$ be the
covering.  

We now proceed to define the labeling.  For $u \in U'$, define the set
of candidate labels as $\calA(u) = \cup_{i=1}^t g_u(\sigma_i)$.  Then
construct the labeling $A$ as follows: for $u \in U'$ let $A(u)$ be a
random label from $\calA(u)$ and for $v \in V$ pick a random $u \in U'$
such that $u \sim v$ and let $A(v) = \phi_{u \rightarrow v}(A(u))$ (if $v$ has no neighbors in $U'$, set $A(v)$ arbitrarily).

The quality of this labeling hinges on \cref{claim:soundness base} below.

\begin{claim}
  \label{claim:soundness base}
  Let $v \in V$ and $u_1, \ldots, u_t \in U'$ be distinct neighbors of $v$
  and write $I_j = \phi_{u_j \rightarrow v}(g_{u_j}(\sigma^j))$. Then, the $I_j$'s are not pairwise disjoint.  
\end{claim}

It is possible that $v$ has fewer than $t$ neighbors in $U'$ but in this case the claim is vacuously true.

\begin{proof}
Suppose for contradiction that the $I_j$'s are pairwise disjoint. By the definition of $I_j$, there exist $\V x^{1,j}, \ldots, \V x^{d,j} \in \sigma_j^{R_U}$ such that
  \begin{enumerate}
  \item $(u_j, \V x^{i,j})$ has color $b$ for all $i \in [d]$, $j \in [t]$.
  \item For all $\beta \not\in I_j$ and $\alpha_j \in \phi_{u_j\rightarrow v}^{-1}(\beta)$ it holds that
    $\{ \V x^{i,j}_{\alpha_j} \}_{i \in [d]} = \sigma^j$.
  \end{enumerate}
From the pairwise disjointness of $I_j$'s, it follows that these strings satisfy \eqref{eq:crosshypedge_cond} for every $\beta \in [L]$ and for all choices of $\alpha_j \in \phi_{u_j \rightarrow v}^{-1}(\beta) \subseteq [R]$ for $j \in [t]$. Thus,
  \[
  \{(u_j, \V x^{i,j})\}_{i \in [d], j \in [t]},
  \]
  forms a hyperedge of $\calH$ which is monochromatic w.r.t.~$\chi$, a contradiction to the fact that $\chi$ was a valid $c$-coloring.
\end{proof}

We also need the following simple claim:

\begin{claim}
\label{claim:dis lb improve}
For any set family $\calS \subseteq 2^{[n]}$ such that no $\Delta$ of them are pairwise disjoint, 
$$\Pr_{s_1, s_2\in \calS}[s_1\cap s_2 \neq \emptyset] \geq \frac{1}{{\Delta-1}}.$$
\end{claim}
\begin{proof}
Define a graph $G(\calS, E)$ on $\calS$ where $s_1 \sim s_2$ if they do not intersect. By the property of $\calS$, $G$ does not contain a clique of size $\Delta$. By Tur{\'a}n's theorem, the number of edges in $G$ is at most
$$|E|\leq \frac{\Delta-2}{\Delta-1}\cdot \frac{|\calS|^2}{2}.$$
Now, the probability that $s_1, s_2\in \calS$ do not intersect is equivalent to saying $(s_1, s_2)\in E$. Thus, the probability is at most 
$$\frac{2|E|}{|\calS|^2} \leq 2\cdot\frac{\Delta-2}{\Delta-1}\cdot \frac{|\calS|^2}{2}\cdot \frac{1}{|\calS|^2} = 1 - \frac{1}{\Delta-1}$$
\end{proof}

Using \cref{claim:soundness base} it is straightforward to obtain a lower
bound on the quality of the randomized labeling.

\begin{claim}
The randomized
  labeling satisfies in expectation at least a
  $\frac{1}{h^2 t^3}$ fraction of the constraints
  between $U'$ and $V$.
\end{claim}

\begin{proof}
 The expected fraction of satisfied constraints involving $v \in V'$ is at least
  \begin{align*}
& \E_{\substack{u_1, u_2 \in U'\\ u_1, u_2 \sim v}}\left [ \Pr_{A(u_1),A(u_2)} [\phi_{u_1 \rightarrow v}(A(u_1)) = \phi_{u_2 \rightarrow v}(A(u_2))] \right] \\
& \quad \quad \quad\ge \E_{\substack{u_1, u_2 \in U\\ u_1, u_2 \sim v}}\left [ \frac{| \phi_{u_1 \rightarrow v}(\calA(u_1)) \cap \phi_{u_2 \rightarrow v}(\calA(u_2))|}{(ht)^2} \right]\\
 & \quad \quad \quad \ge \frac{1}{(ht)^2} \Pr_{\substack{u_1, u_2 \in U\\ u_1, u_2 \sim v}}\left [ \phi_{u_1 \rightarrow v}(\calA(u_1)) \cap \phi_{u_2 \rightarrow v}(\calA(u_2)) \ne \emptyset \right]\\
& \quad \quad \quad \ge\frac{1}{(ht)^2}\cdot\frac{1}{t}
  \end{align*}
where the last inequality follows from \cref{claim:soundness base} and \cref{claim:dis lb improve}.
\end{proof}

To summarize, the constructed labeling satisfies a
$\frac{1}{h^2 t^3}\cdot\frac{1}{c T}$ fraction of all constraints between the $U$ and $V$, and we are done.
\end{proof}

\begin{proof}[Proof of \cref{thm:large q}]
  The proof follows from Lemma~\ref{lemma:completeness_gen} and Lemma~\ref{lemma:soundness_gen} and by setting $r$ such that the soundenss of the Label Cover is $2^{-\Omega(r)} \ll \frac{1}{ d^4c^3 t^3 2^{td\log q}}$. 
\end{proof}

\section{Concluding Remarks}
\label{sec:conclusion}

We have shown improved hardness of finding $2$-colorings in rainbow
colorable hypergraphs, and of finding $c$-colorings of almost rainbow
colorable hypergraphs.  There are a number of interesting open
questions.  For the $\Rainbow$ problem, the smallest open case is
currently $\Rainbow(5,4,2)$.  For various reasons our methods are
insufficient to tackle this problem, and it would be interesting to
know whether this problem is $\np$-hard or not.

On the combinatorial side, our analysis of the hypergraph gadgets
$\KHn{\lfloor d/2\rfloor}{n}{[d]}$ only yield non-$2$-colorability
(\cref{lemma: ch 3t2t}) and the only upper bounds we have on the size
of independent sets in those graphs is the trivial $1-1/d$ whereas we
believe the true answer should be $1/2-o(1)$ (which immediately
implies non-$2$-colorability).  Such a bound would not help in
improving our hardness results but would still be interesting to
understand.

In some sense, the reason why we only get hardness for $2$-colorings
is that the soundness argument contains steps along the following
lines: (i) no cloud can be almost monochromatic, (ii) therefore since
there are only two colors, each cloud contains a constant fraction of
vertices of each color, (iii) in order for the randomized labeling to
fail, the involved clouds would need to have a very small fraction of
vertices of some color.  Here, step (ii) is clearly not true for
colorings with more than $2$ colors.

\paragraph{Acknowledgements}

We would like to thank Mike Saks for the helpful discussions in the early stages of this work. AP would like to thank Jeff Kahn for the very helpful discussions. We would also like to thank Florian Frick for bringing~\cite{LZ07},~\cite{ACCFS18} , and the restrictions for Sarkaria's theorem to our notice.

\section*{Appendix}
\appendix
% !TeX root = rainbow.tex

\section{The \texorpdfstring{$\Rainbow\left(td+\lfloor \nicefrac{d}{2}\rfloor, t(d-1)+1, 2\right)$}{rianbow}-hardness}
\label{sec:simplegen}

In this section, we give a generalization of the $\Rainbow(4,3,2)$ result from the Section~\ref{sec:four}. This gives an elementary proof of $\Rainbow(td+\lfloor \nicefrac{d}{2}\rfloor, t(d-1)+1, 2)$-hardness.

\begin{theorem}[\cref{thm:simple} restated]
  For every $t \ge 1$ and $d\geq 2$, $\Rainbow(td+\lfloor \nicefrac{d}{2}\rfloor, t(d-1)+1, 2)$ is $\np$-hard.
\end{theorem}

In the proof of this theorem, we use the $c=2$ case of the covering bound \cref{thm:covgen} (c.f.~\cref{sec:covering bound}).  While we are not aware of any non-topological proof of the full version of \cref{thm:covgen}, the $c=2$ case does admit an simple inductive proof, provided here for completeness.

\begin{lemma}[$c=2$ case of \cref{thm:covgen}]
\label{lemma:two}
For every $q \ge d \ge 2$, $B(q, d, 2) = \lceil \frac{q-1}{d-1} \rceil$, i.e., for every $f: {[q] \choose d} \rightarrow \{0,1\}$, there are $b = \lceil \frac{q-1}{d-1} \rceil$ sets $S_1, \ldots, S_b \in {[q] \choose d}$ such that $\cup S_i = [q]$ and $f$ is constant on $S_1, \ldots, S_b$.
\end{lemma}
\begin{proof}
We prove it by induction on $q$. The base case when $q=d$ is trivial.  Let $q \ge 2d-1$. If $f$ is not a constant function then there exists $T\in {[q]\choose d-1}$ and $i,j \in [q]\setminus T$, such that $f(T\cup \{i\}) \neq f(T\cup \{j\})$. By induction, for the restricted function $\tilde{f} : {[q]\setminus T \choose d} \rightarrow \{0,1\}$, there exists a cover $\tilde{\mathcal{S}}\subseteq {[q]\setminus T \choose d}$ of $[q] \setminus T$ such that $\tilde{f}$ is constant on $\tilde{\mathcal{S}}$ and $|\tilde{\mathcal{S}}| \leq \lceil \frac{q-1-(d-1)}{d-1} \rceil \leq \lceil \frac{q-1}{d-1} \rceil - 1$. Either $\mathcal{S} = \tilde{\mathcal{S}} \cup \{T\cup \{i\}\}$ or $\mathcal{S} = \tilde{\mathcal{S}} \cup \{ T\cup \{j\}\}$ gives the required covering whose size is at most $\lceil \frac{q-1}{d-1} \rceil$.

The remaining case $d < q \le 2d-2$ is handled similarly -- in this case
we take $T \in {[q] \choose q-d}$ in order to end up in the base case
and get a cover of size $2$, as desired.
\end{proof}

\subsection{Reduction}
We are now ready to give the reduction. We start with a multi-layered Label Cover $\mathcal{L}$ instance with parameters $\ell$ and $r$ to be determined later. We reduce it to the hypergraph $\calH(\calV, \calE)$. The reduction is given in \cref{fig:red3t2t}.

\begin{figure}[ht!]
\begin{tcolorbox}[
    standard jigsaw,
    opacityback=0, 
    ]
Let $q\defeq t(d-1)+1$, where $t\geq1$ and $d\geq 2$ are integers.
\paragraph{Vertices $\calV$.} Each vertex $v$ from layer $i$ in the layered Label Cover instance $\mathcal{L}$ is replaced by a cloud of size $q^{R_i}$ denoted by $C[v] := \{v\}\times [q]^{R_i}$. We refer to a vertex from the cloud $C[v]$ by a pair $(v,\V x)$ where $\V x\in [q]^{R_i}$. The vertex set of the hypergraph is given by 
$$ \calV = \cup_{v\in \cup_i X_i} C[v].$$

\paragraph{Hyperedges $\calE$.} There are two types of edges.
\begin{description}
\item[Type 1:] For every $1\leq \inda < \indb \leq \ell$, every vertex $v\in X_\indb$ and every set of $t$ neighbors $u_1, u_2, \ldots, u_t$ of $v$ from layer $X_\inda$, we add the following hyperedges.
Let $\pi_i = \phi_{u_i \rightarrow v}$ be the projection constraint between $u_{i}$ and $v$ for $1\leq i\leq t$.

Let $\V x^{i,j} \in [q]^{R_\inda}$ be a set of $t d$ strings indexed by $i \in [d]$ and $j \in [t]$, let $\V y^{i} \in [q]^{R_\indb}$ be $\lfloor \nicefrac{d}{2} \rfloor$ strings indexed by $i \in [\lfloor \nicefrac{d}{2} \rfloor]$.  If it holds that for every $\beta \in [R_\indb]$ and all choices of $\alpha_j \in \pi_j^{-1}(\beta) \subseteq [R_\inda]$ for $j \in [t]$ that
\begin{equation}
\label{eq:crosshypedge_cond 3t2t}
\left\{ \V x^{i,j}_{\alpha_j} \,|\,i \in [d], j \in [t] \right\} \,\bigcup\, \left\{ \V y^{i}_\beta \,|\, i \in [\lfloor \nicefrac{d}{2} \rfloor] \right\} = [q],
\end{equation}
then we add add the hyperedge
\[
\left\{(u_j, \V x^{i,j})\right\}_{i \in [d], j \in [t]} \,\bigcup\, \left\{(v, \V y^{i}) \right\}_{i \in [\lfloor \nicefrac{d}{2} \rfloor]} \in \binom{\calV}{td + \lfloor \nicefrac{d}{2} \rfloor}
\]
to the hypergraph.

\item[Type 2:] For every $1\leq \eta \leq \ell$, $v\in X_\eta$, in the cloud $C[v]$, add a hyperedge $\{\V y^1, \V y^2, \ldots, \V y^{td+\lfloor \nicefrac{d}{2} \rfloor}\}$  if for all $\beta \in [R_\eta]$
$$ \left\{ \V y^{i}_\beta \,|\, i \in [ td+\lfloor \nicefrac{d}{2} \rfloor] \right\} = [q]$$
\end{description}
\end{tcolorbox}
\caption{Reduction to  $\Rainbow(td+\lfloor \nicefrac{d}{2}\rfloor, t(d-1)+1, 2)$. }
\label{fig:red3t2t}
\end{figure}

For comparison with the warmup reduction \cref{fig:red3} for $\Rainbow(4,3,2)$, observe that if we set $t=1$, $d=3$, and only take the Type 1 edges from \cref{fig:red3t2t}, we obtain the same reduction.  The sole purpose of the additional Type 2 edges used in this more general reduction is to force any $2$-coloring of the resulting hypergraph to be somewhat balanced within each cloud (see further \cref{claim: ind 3t2t} below).  In the $\Rainbow(4, 3,2)$ case this was instead achieved via \cref{fact:khn_indset}.

\subsection{Analysis}
\begin{lemma}[Completeness]
If the Label Cover instance is satisfiable then the hypergraph $\calH$ is $q$-rainbow colorable.
\end{lemma}
\begin{proof}
Let $A: \bigcup_i X_i \rightarrow \bigcup_i [R_i]$ define the satisfiable labeling to the layered Label Cover instance. The rainbow $q$-coloring of the hypergraph is given by assigning a vertex $(v,\V x)$ with a color $\V x_{A(v)}$. 

To see that this is a rainbow $q$-rainbow coloring, consider any Type $1$ hyperegde in the hypergraph between the clouds $C[u_1], C[u_2], \ldots, C[u_t]$ and $C[v]$ where $u_1, u_2, \ldots, u_t \in X_\inda$ and $v\in X_\indb$. This hyperegde is of the form 
\[
\{(u_j, \V x^{i,j})\}_{i \in [d], j \in [t]} \,\bigcup\, \{(v, \V y^{i}) \}_{i \in [\lfloor \nicefrac{d}{2} \rfloor]} \in {\calV \choose td + \lfloor \nicefrac{d}{2} \rfloor}
\]
satisfying \eqref{eq:crosshypedge_cond 3t2t}. By definition, $\chi$ assigns color $\V x^{i, j}_{A(u_j)}$ to vertices $\{(u_j, \V x^{i,j})\}$ for $i\in [d]$ and $j\in [t]$ and  $\V y^{i}_{A(v)}$  to $(v, \V y^{i})$ for $i\in [\lfloor \nicefrac{d}{2} \rfloor]$. It is easy to see from \eqref{eq:crosshypedge_cond 3t2t} that these vertices get $q$ distinct colors since $A(u_j) \in \pi_{j}^{-1}(A(v))$ for all $1\leq j\leq t$. 

Also, all Type $2$ hyperedges trivially contain all the $q$ colors. Hence $\chi$ is a valid $q$-rainbow coloring.
\end{proof}

We now prove the main soundness lemma.
\begin{lemma}[Soundness]
\label{lemma:soundness_gen2}
If $\ell \ge 8\cdot (td)^{2td}$ and $\calH$ is properly $2$-colorable then there is an assignment $A$ to the layered Label Cover instance which satisfies an $2^{-O(t^2d^2)}$ fraction of all constraints between some pair of layers $X_i$ and $X_j$.
\end{lemma}

In particular setting the layered Label Cover parameter $r \gg t^2d^2$ in Theorem~\ref{theorem:multiLC}, proves \cref{thm:simple}.

\begin{proof}
Assume for contradiction that the hypergraph $\calH$ is $2$-colorable. Fix a $2$-coloring $\chi : \calV \rightarrow \{0,1\}$ of the vertices of $\calH$.

We have a following simple claim about the upper bound on the density of a color class in every cloud.

\begin{claim}
\label{claim: ind 3t2t}
For every $1\leq \eta \leq \ell$, $v\in X_\eta$ and $b\in \{0,1\}$, in the cloud $C[v]$, the fraction of vertices colored with color $b$ is at least $1/q$.
\end{claim}
\begin{proof}
Consider the class of shifts of $\V x\in [c]^{[R_\eta]}$ defined as $[\V x] := \{\V x +\V 1,\V x+\V 2,  \ldots, \V x+\V q\}$, where $+$ is coordinate-wise addition (modulo $q$). Suppose for contradiction that the fraction of vertices in $C[v]$ that are colored $b$ is less than $1/q$. Thus, there exists $\V x$ such that $[\V x]$ is monochromatic with color $1-b$. Since at least $1-1/q$ fraction of $C[v]$ is colored with color $1-b$, there exist a set of distinct strings $\V y^1, \V y^2, \ldots, \V y^{(t-1)+\lfloor \nicefrac{d}{2}\rfloor} \notin [\V x]$, such that $\chi(\V y^i) = 1-b$ for all $i\in [(t-1)+\lfloor \nicefrac{d}{2}\rfloor]$. But then $\{ \V y^i \mid i\in [(t-1)+\lfloor \nicefrac{d}{2}\rfloor]\} \cup [\V x]$ is a hyperedge of Type $2$ in $\calH$ which is monochromatic w.r.t.~the coloring $\chi$.
\end{proof}

%item
For every $u\in X_i$, define functions $f_u : {[q]\choose d} \rightarrow \{0,1\}$, and $g_u:{[q] \choose d} \rightarrow {[R_i]\choose \leq d}$ as follows. For a $\sigma \in {[q]\choose d}$, in a cloud $C[u]$, consider the induced $d$-uniform hypergraph $\KHn{{\lfloor \nicefrac{d}{2} \rfloor}}{R_i}{\sigma}$. Look at the coloring on these vertices induced by $\chi$ i.e.~$\chi_{u, \sigma} : \sigma^{R_i} \rightarrow \{0,1\}$ defined by $\chi_{u, \sigma} (\V x) = \chi((u,\V x))$.  By \cref{lemma: ch 3t2t}, there exists a color class, say $b\in \{0,1\}$, such that there exists a monochromatic hyperedge with color $b$ in $\KHn{{\lfloor \nicefrac{d}{2} \rfloor}}{R_i}{\sigma}$.  Set $f_u(\sigma) = b$, where $b$ is one such color class, breaking ties arbitrarily. Also, set $g_u(\sigma) = J_u$ if $J_u \subseteq  [R_i]$ is the set of {\it noisy coordinates} in the $b$-monochromatic hyperedge, again breaking ties arbitrarily. If none of the coordinates are noisy in the hyperedge, then set $g_u(\sigma) = \{1\}$.

%\item 
By Lemma~\ref{lemma:two}, there exist for each variable $u$ subsets $\sigma^u_1, \sigma^u_2, \ldots, \sigma^u_t \in {[q] \choose d}$ and a color $b_u \in \{0,1\}$ such that $f_u(\sigma^u_j) = b_u$ for all $j \in [t]$ and $\cup_{j=1}^t \sigma^u_j = [c]$.    Write $S_u = (\sigma^u_1,\ldots \sigma^u_t) \in {[q] \choose d}^t$.  Next, associate each layer $i$ with the most frequent value among $(S_u, b_u)$ over all vertices $u \in X_i$.  For each layer $i \in [\ell]$, let $\tilde{X}_i$ be the set of vertices in $X_i$ with the same label as layer $i$.

%\item 
Let $T$ be the total number of coverings of ${[q]\choose d}$ of size at most $t$.  A trivial upper bound on $T$ is ${q \choose d}^t \leq (td)^{td}$. Since $\ell \geq 8\cdot (td)^{2td} \ge 8T^2$, there exists $m = 4T$ layers which are all associated with the same pair $(S,b)$, and in each of these $4T$ layers, at least a $1/(2T) = 2/m$ fraction of all variables are associated with $(S, b)$.  By the weak density property of the Label Cover instance, it follows that there exist two layers $i$ and $j$ such that the fraction of constraints between $\tilde{X}_i$ and $\tilde{X}_j$ is at least a $\frac{1}{16T^2}$ fraction of all constraints between $X_i$ and $X_j$.

For the rest of the analysis, we set $U = \tilde{X}_i$ and
$V = \tilde{X}_j$ and focus on satisfying the constraints between $U$
and $V$.  Let $S = \{\sigma_1, \sigma_2, \ldots, \sigma_t\}$ be the
covering.  

{\bf Labeling: }We now proceed to define the labeling.  For $u \in U$, define the set
of candidate labels as $\calA(u) = \cup_{i=1}^t g_u(\sigma_i)$.  Then
construct the labeling $A$ as follows: for $u \in U$ let $A(u)$ be a
random label from $\calA(u)$ and for $v \in V$ pick a random $u \in U$
such that $u \sim v$ and let $A(v) = \phi_{u \rightarrow v}(A(u))$.

To analyze the quality of the labeling, we need the following two claims, which together form a generalization of the simpler \cref{claim:setsize} used in the $\Rainbow(4,3,2)$ reduction -- that if the neigbors $u \in U$ of $v \in V$ suggest many incompatible candidate labels for $v$, then a large fraction of vertices $(v, \V y)$ in $C[v]$ must not have color $b$ (contradicting \cref{claim: ind 3t2t}).

\begin{claim}
  \label{claim:soundness base3t2t}
  Let $v \in V$ and let $u_1, \ldots, u_t \in U$ be distinct neighbors of $v$
  and let $I_j = \phi_{u_j \rightarrow v}(g_{u_j}(\sigma_j))$.  Let
  $I = \cup_{j=1}^t I_j$ and suppose that the $I_j$'s are all pairwise
  disjoint.
  Then there exists a string $\V w \in [q]^I$ such that for all
  $\V y \in [q]^{R_V}$ with $\V y_{|I} = \V w$, the vertex
  $(v, \V y)$ does not have the color $b$.
\end{claim}

\begin{proof}
For all $j\in [t]$, by definition of $I_j$, there exist $\V x^{1,j}, \ldots, \V x^{d,j} \in \sigma_j^{R_U}$ such that
  \begin{enumerate}
  \item $(u_j, \V x^{i,j})$ has color $b$ for all $i \in [d]$, $j \in [t]$.
  \item There exists $J_{u_j} \subseteq [R_U]$, $\phi_{u_j \rightarrow v}(J_{u_j}) = I_j$ such that for all $\alpha\notin J_{u_j}$ it holds that
    $\{ \V x^{i,j}_{\alpha} \}_{i \in [d]} $ $= \sigma^j$ and for all  $\alpha \in J_{u_j}$, we have $|\{\V x^{i,j}_{\alpha} \}_{i \in [d]}| \geq \lceil\nicefrac{d}{2}\rceil$. Moreover, there exists a subset $S_{u_j} \subseteq \sigma_j$ of size at least $\lceil d/2\rceil$ such that for all  $\alpha \in J_{u_j}$, the set $\{\V x^{i,j}_{\alpha} \}_{i \in [d]}$ contains all the elements from $S_{u_j}$.
    \end{enumerate}
  Consider any set of $\lfloor d/2\rfloor$ strings $\V y^1, \ldots, \V y^{\lfloor d/2\rfloor} \in [q]^{R_V}$ such that for all $\beta \in I_j$ it holds that 
 \begin{equation}
\label{eq:fulfill}
  \{ \V y^{i}_{\beta} \}_{i \in [\lfloor d/2\rfloor]} \supseteq \sigma_j \setminus S_{u_j}.
\end{equation}
Note that $|\sigma_j \setminus S_{u_j}|$ is at at most $\lfloor d/2\rfloor$ and hence there are $\V y^1, \ldots, \V y^{\lfloor d/2\rfloor} \in [q]^{R_V}$ satisfying \eqref{eq:fulfill} for all $j\in [t]$.  By construction it follows that these strings along with $\{\V x^{i,j}\}_{i\in [d], j\in [t]}$ satisfy \eqref{eq:crosshypedge_cond 3t2t} and thus
  \[
  \{(u_j, \V x^{i,j})\}_{i \in [d], j \in [t]} \cup \{(v, \V y^{i}) \}_{i \in [\lfloor d/2\rfloor]},
  \]
  forms a hyperedge of $\calH$.  It follows that at least one of
  $(v, \V y^i)$ must have a  color than different $b$. Let $H\subseteq [\lfloor d/2\rfloor]$ be the set of indices $i$ such that $(v, \V y^i)$ is not colored $b$.

Suppose for the sake of contradiction, for all such $(v, \V y^i)$ which is $not$ colored $b$, there exists a string $\V z^i$ agreeing with $\V y^i$ at locations $I$ i.e. $\V y^i_{|I} = \V z^i_{|I}$ such that the color of vertex $(v, \V z^i)$ is $b$. One can check that  $\{(u_j, \V x^{i,j})\}_{i \in [d], j \in [t]}\cup \{(v, \V z^{i}) \}_{i \in H} \cup \{(v, \V y^{i}) \}_{i \in [\lfloor d/2\rfloor]\setminus H}$  is a valid hyperedge with color $b$, a contradiction. Therefore there exists $i\in T$ such that for all strings $\V y \in  [q]^{R_V}$ with $\V y_{|I} = \V y^i_{|I}$, the vertex $(v, \V y)$ does not have color $b$.
\end{proof}

The following claim {\em rules out} that for many neighbors of $v$, the collection of candidate labelings $\phi_{u\rightarrow v}(\calA(u))$ are pairwise disjoint.

\begin{claim}
  \label{claim:B claim 3t2t}
  Let $B =  t \cdot q^{td}\cdot\ln q$ and $v \in V$.  Then for any $B$ distinct neighbors $u_1, \ldots, u_B \in U$ of $v$, it holds that the label sets
  \[
  \phi_{u_j\rightarrow v}(\calA(u_j)),
  \]
  for $j \in [B]$ are not all pairwise disjoint.
\end{claim}

\begin{proof}
  Suppose for contradiction that $B$ such neighbors exist where
  the corresponding label sets are all pairwise disjoint.  Split them
  into $D := B/t$ groups of size $t$.  By \cref{claim:soundness
    base3t2t} it follows that there exist $D$ disjoint label sets
  $I_1, \ldots, I_D \subseteq [R_V]$ and strings
  $\V w^1 \in [q]^{I_1}, \ldots, \V w^D \in [q]^{I_D}$ such that
  $(v, \V y)$ does not have color $b$ whenever $\V y_{|I_j} = \V w^j$
  for some $j \in [D]$.  Furthermore the sets $I_j$ have size at most $|I_j| \le td$ so there at most a fraction $1 - q^{-td}$ of strings in $[q]^{R_V}$ differ from $\V w^j$ on $I_j$.  By the disjointness of the $I_j$'s we thus have that the total fraction of vertices in the cloud $C[v]$ that have color $b$ is at most
  \[
  (1-q^{-td})^D \le e^{-\frac{D}{q^{td}}}.
  \]
  However, by \cref{claim: ind 3t2t}, for every
  $v \in V$ the cloud
  $C[v]$ must contain at least a fraction $\frac{1}{q}$ of the vertices with color $b$. Therefore, it follows that we must
  have $D/q^{td} \le \ln q$ and the claim follows.
\end{proof}

Using \cref{claim:B claim 3t2t} it is straightforward to obtain a lower
bound on the quality of the randomized labeling.

\begin{claim}
  Let $B = t \cdot q^{td}\cdot\ln q$ be as in \cref{claim:B claim 3t2t}.  Then the randomized
  labeling satisfies in expectation at least a
  $\left(\frac{1}{t^2B} \right)$ fraction of the constraints
  between $U$ and $V$.
\end{claim}

\begin{proof}
 The expected fraction of satisfied constraints involving $v \in V$ is at least
  \begin{align*}
&\E_{\substack{u_1, u_2 \in U\\ u_1, u_2 \sim v}}\left [ \Pr_{A(u_1),A(u_2)} [\phi_{u_1 \rightarrow v}(A(u_1)) = \phi_{u_2 \rightarrow v}(A(u_2))] \right] \\
& \quad \quad \quad\ \ge \E_{\substack{u_1, u_2 \in U\\ u_1, u_2 \sim v}}\left [ \frac{| \phi_{u_1 \rightarrow v}(\calA(u_1)) \cap \phi_{u_2 \rightarrow v}(\calA(u_2))|}{t^2} \right]\\
 & \quad \quad \quad\ \ge \frac{1}{t^2} \Pr_{\substack{u_1, u_2 \in U\\ u_1, u_2 \sim v}}\left [ \phi_{u_1 \rightarrow v}(\calA(u_1)) \cap \phi_{u_2 \rightarrow v}(\calA(u_2)) \ne \emptyset \right]\\
& \quad \quad \quad\ \ge\frac{1}{t^2}\cdot\frac{1}{B}
  \end{align*}
where the last inequality follows from \cref{claim:B claim 3t2t} and \cref{claim:dis lb improve}.
\end{proof}

Thus, the constructed labeling satisfies a
$\frac{1}{B}\cdot \left(\frac{1}{2 T  t}\right)^2 = \frac{1}{tq^{td}\ln q}\frac{1}{4(td)^{2td} t^2} \geq2^{-O(t^2d^2)}$ fraction of all constraints between the two layers, and this finishes the proof.
\end{proof}

\subsection{Proof of \texorpdfstring{\cref{cor: 2sqrtk}}{1.4}}
We start with the following simple claim:

\begin{claim}
\label{claim:incdecuniformity}
If $\Rainbow( k, q, 2)$ is $\np$-hard then $\Rainbow(k+1 , q, 2)$ is $\np$-hard.
\end{claim}
\begin{proof}
Let $H(V, E)$ be an instance of $\Rainbow( k, q, 2)$ . Construct a $k+1$ uniform hypergraph $H_1(V_1, E_1)$ as follows: $V_1 = V \cup \{v_1, v_2, \ldots, v_{k+1}\}$ where $\{v_1, v_2, \ldots, v_{k+1}\}$  are the extra set of vertices not in $V$. For every hyperedge $e\in E$ add $(e\cup v_i)$ to $E_1$ for all $1\leq i\leq k+1$. Also add  $\{v_1, v_2, \ldots, v_{k+1}\}$  to $E_1$. This finishes the reduction. Now, if $H$ is $q$-rainbow colorable, then coloring $\{v_1, v_2, \ldots, v_{k+1}\}$  with $q$ different colors and keeping the colors of vertices $V$ as given by the $q$-rainbow coloring of $H$ gives a $q$-rainbow coloring of $H_1$. On the other hand, if $H_1$ is $2$-colorable then the restriction of the $2$-coloring to $V$ gives a proper $2$-coloring of $H$.

\end{proof}

\begin{proof}[Proof of \cref{cor: 2sqrtk}]
  Let $t = \left\lfloor \frac{1}{2} \sqrt{k} \right\rfloor$ and set
  $d$ to be the largest integer such that
  $u := td + \lfloor d/2\rfloor \le k$.  Observe that
  $d \le 2\sqrt{k}$ and that $k-u \le t+1$.  Applying
  \cref{thm:simple} and $k-u$ repetitions of
  \cref{claim:incdecuniformity}, we have that $\Rainbow(k, q, 2)$ is
  $\np$-hard for $q = t(d-1)+1 = u - \lfloor d/2 \rfloor - t + 1$.  The difference between $k$ and $q$ is
  \[
  k-q = k - u + \lfloor d/2 \rfloor + t - 1 \le \lfloor d/2 \rfloor + 2t \le 2 \lfloor \sqrt{k} \rfloor.
  \]
\end{proof}

\bibliographystyle{alpha}
\bibliography{rainbow}

\end{document}